\documentclass[12pt]{amsart}
\usepackage{latexsym}
\usepackage{a4}
\usepackage{amssymb}
\usepackage{amsbsy}
\usepackage{bm}
\usepackage{algpseudocode} 
\usepackage{enumitem}
\newtheorem{thm}{Theorem}[section] 
\newtheorem{pro}[thm]{Proposition} 
\newtheorem{lem}[thm]{Lemma}

\theoremstyle{definition}
 
\newtheorem{prb}[thm]{Problem}
\newtheorem{de}[thm]{Definition}
 
\numberwithin{equation}{section} 

\newcommand{\ob}[1]{{\mathbb{#1}}}

\newcommand{\Mat}{\operatorname{Mat}}
\newcommand{\MatThree}[9]
           { \left( \begin{smallmatrix}
                           {#1} & {#2} & {#3} \\
                           {#4} & {#5} & {#6} \\
                           {#7} & {#8} & {#9}
             \end{smallmatrix} \right) }
\newcommand{\VecTwo}[2]{ 
   \left( 
   \begin{smallmatrix} 
      #1 \\ #2 
   \end{smallmatrix} 
   \right) 
   } 
\newcommand{\VecTwoBig}[2]{ 
   \left( 
   \begin{array}{c} 
      #1 \\ #2 
   \end{array} 
   \right) 
}

\newcommand{\N}{\Bbb{ N}} 
\newcommand{\Z}{\Bbb{ Z}} 

\newcommand{\ul}[1]{\underline{#1}}

\newcommand{\MatTwo}[4]
           { \left( \begin{smallmatrix}
                           {#1} & {#2} \\
                           {#3} & {#4}
             \end{smallmatrix} \right) }

\DeclareMathAlphabet{\mathbfsl}{OT1}{cmr}{bx}{it}
\newcommand{\tupBold}[1]{\mathbfsl{#1}}
\newcommand{\vb}[1]{\tupBold{#1}}

\newcommand{\NP}{\mathrm{NP}}
\newcommand{\PP}{\mathrm{P}}

\newcommand{\matrixRing}{\operatorname{Mat}_m(\baseField)}
\newcommand{\matrixRingZZ}{\operatorname{Mat}_2(\ob{F}_2)}
\newcommand{\matrixGroup}{\operatorname{GL}_m(\baseField)}
\newcommand{\vectorSpace}{\baseField^m}
\newcommand{\Holomorph}{\operatorname{Hol}(q,m)}
\newcommand{\baseField}{\mathbb{F}_q}
\newcommand{\length}[1]{\lvert #1\rvert}
\newcommand{\setsize}[1]{\lvert #1 \rvert}

\newcommand{\PolSat}{\textsc{PolSat}}
\newcommand{\PolEqv}{\textsc{PolEqv}}
\newcommand{\GL}{\operatorname{GL}}
\def\GFF{\mathbb{F}_4}
\newcommand{\sgn}{\operatorname{sgn}}

\title[On solving equations over $S_4$]{On the complexity of solving equations over the symmetric group
  $S_4$}
 
\author{Erhard Aichinger \and Simon Gr\"unbacher}

\address{
Institut f\"ur Algebra,
Johannes Kepler Universit\"at Linz,
4040 Linz,
Austria}
\email{\tt erhard@algebra.uni-linz.ac.at}
\email{simon.gruenbacher@gmail.com}


\urladdr{http://www.jku.at/algebra}
\thanks{Supported by the Austrian Science Fund (FWF):P33878}

\date{\today}

\begin{document}
\bibliographystyle{amsalpha}
\begin{abstract}
The complexity of solving equations over finite groups has been an active area of research over the last two decades, starting with \cite{GR:TCOS}.
One important case of a group with unknown complexity is the symmetric group $S_4.$
In \cite{IKK:IPIM}, the authors prove $\exp(\Omega(\log^2 n))$ lower bounds for the satisfiability and equivalence problems over $S_4$ under the Exponential Time Hypothesis.
In the present note, we prove that the satisfiability problem $\textsc{PolSat}(S_4)$ can be reduced to the equivalence problem $\textsc{PolEqv}(S_4)$ and thus, the two problems have the same complexity.
We provide several equivalent formulations of the problem. 
In particular, we prove that $\textsc{PolEqv}(S_4)$ is equivalent to the circuit equivalence problem for $\operatorname{CC}[2,3,2]$-circuits,
which were introduced
in \cite{IKK:COMC}.
    Under the strong exponential size hypothesis (see \cite{IKK:IPIM}), such circuits cannot compute $\operatorname{AND}_n$ in size $\exp(o(\sqrt{n})).$
Our results provide an upper bound on the complexity of $\textsc{PolEqv}(S_4)$ that is based on the minimal size of $\operatorname{AND}_n$ over 
$\operatorname{CC}[2,3,2]$-circuits.
\end{abstract}

\maketitle
\section{Introduction} \label{sec:intro}

The \emph{polynomial satisfiability problem} $\textsc{PolSat}(G)$ over a finite group $G$ asks whether a given equation over $G$ has a solution. 
For example, $x_1x_2(1\; 2) = x_2x_1$ is an unsatisfiable equation over the symmetric group $S_4.$
The study of this problem has been initiated by \cite{GR:TCOS}, where the problem shown to be in $\textrm{P}$ for nilpotent groups and $\textrm{NP}$-complete for nonsolvable groups.
The problem was later determined to be in \textrm{P} also for $S_3$ \cite{HS:TCOC} and $A_4$ \cite{HS:EAES} as well as for other groups that have a certain semidirect product representation \cite{FH:TCOT}.

In contrast to $\textsc{PolSat}(G),$ the \emph{polynomial equivalence problem} $\textsc{PolEqv}(G)$ asks whether a polynomial equation holds for all variable assignments.
Clearly $\textsc{PolSat}(G) \in \textrm{P}$ implies $\textsc{PolEqv}(G) \in \textrm{P},$ since $p(x) = 1$ is not valid if and only if $p(x) = a$ is satisfiable for some $a\in G\setminus \{1\}.$
Assuming the Exponential Time Hypothesis, the converse implication does not hold in general: In 
\cite{IKK:COMC}, the authors show that for the dihedral group $D_{15},$ the equivalence problem can be solved in polynomial time whereas the satisfiability problem cannot unless the Exponential Time Hypothesis fails.
The Exponential Time Hypothesis (ETH) was introduced in \cite{IP:OTCO} and states that 3-\textsc{SAT} requires time $\exp(cn)$ for some $c > 0.$

Equation satisfiability over solvable groups is related to the size of $\text{CC}$-circuits that compute $n$-bit conjunction.
A $\operatorname{CC}[m_1, \dots, m_d]$-circuit is a depth $d$ circuit
with input and output alphabet $\{0,1\}$, which we see as a subset
of $\N_0$. Its unbounded fan-in gates at depth $h$ output $1$ if their inputs sum to a multiple of $m_h$. %
We assume that some input wires are set to the constant $1$. 
Such circuits were introduced in \cite{IKK:COMC}; a general introduction to
circuits is \cite{Vo:ITCC}, and results relating algebras to circuits can also found, e.g., in \cite{Ko:CATE}.
The \emph{size} $\length{C}$ of such a circuit $C$ is defined by the number of gates.
In \cite{ST:FAAC}, the authors show that such modular
circuits efficiently recognize the same languages as nonuniform deterministic finite automata (NUDFA) over solvable groups.
Applying this construction to the language $L = \{1^n\}$ recognized by $\operatorname{AND}_n$ and the group $G = S_4,$ we can see that a $\gamma(n)$ lower bound for conjunction on some subclass of modular circuits 
leads to a $\gamma(n)^{c}$ lower bound for the size of $S_4$-automata that compute conjunction.
Using the proof technique in \cite[Theorem~2]{BMT:ESAP}, this yields an $\exp(O(\log(n)\gamma^{-1}(n^d)))$ time algorithm to solve $\textsc{ProgramSat}(S_4),$ and therefore also $\textsc{PolSat}(S_4)$ using the reduction described in Lemma 1 of \cite{BMT:ESAP}. A detailed account is given in Theorem~\ref{thm:algo} and Section~\ref{sec:circuits}.

The Strong Exponential Size Hypothesis (SESH), introduced in \cite{IKK:IPIM}, states that for distinct primes $p_1, \dots, p_h,$ a $\operatorname{CC}[p_1, \dots, p_h]$ circuit computing $n$-bit conjunction has size at least $\exp(c\cdot n^{1/(h-1)})$ for some $c > 0.$
It is known \cite{CW:SACF} that for each $\delta > 0,$ there exists $m\in\N,$ a prime $p$ and a family of $\operatorname{CC}[p, m, p]$-circuits of size $\exp(O(n^\delta))$ which compute any symmetric function, including $\operatorname{AND}_n.$
Furthermore, by \cite{IKK:COMC} there exist $\operatorname{CC}[2,3,2]$-circuits of size $\exp(O(\sqrt{n}))$ that compute $\operatorname{AND}_n.$ %
Since the complexity of computing conjunction varies strongly based on the exact type of CC-circuit, it is important which type corresponds to $\textsc{PolSat}(S_4).$
In the present note, we resolve this question.

In \cite{IKK:IPIM}, the authors show that both $\textsc{PolSat}(S_4)$ and $\textsc{PolEqv}(S_4)$ cannot be solved in time $\exp(c\cdot \log(n)^2)$ for any $c>0$ unless ETH fails.
Note that $n$ does not refer to the number of variables in the equation, but rather to the number of symbols.
Formally, the input to both problems is a \emph{group word} $w\in (G \cup X)^n$ of length $n$ and we want to know whether this word evaluates to $1$ for some (resp. all) assignments $a: X \rightarrow G.$

In this paper, we show that $\textsc{PolSat}(S_4)$ can be reduced to the complement of $\textsc{PolEqv}(S_4)$. 
This shows that the complexity cannot differ like it does in the case of dihedral groups.
We further show that the problems can be reduced to the equivalence problem for $\operatorname{CC}[2,3,2]$-circuits.

The structure of all the problems that are investigated here is the
following: We are given an expression $e (x)$ that contains some variables
$x$ and evaluates to some element in a domain $D$ for each setting
of the variables. We suppose that $z$ is a distinguished element
from $D$. Then we ask whether the expression $e(x)$ evaluates to $z$
for all $x$ (the \emph{equivalence problem}) or whether there exists
an $x$ such that $e(x)$ evaluates to $z$ (the \emph{satisfiability problem}).
Hence an equivalence problem asks whether $\forall x : e (x) = z$
and a satisfiability problem asks whether $\exists x : e (x) = z$.
Some of our reductions will reduce to the \emph{complement of the
equivalence problem}, which asks whether $\exists x : e (x) \neq z$. For example,
for solving the equation $f(x_1,\ldots,x_n) = 0$ over the finite field
$\ob{F}_q$, we might ask whether there are $x_1, \ldots, x_n$
with $f(x_1, \ldots, x_n)^{q-1} \neq 1$. Then any $\vb{x}$ witnessing
$\exists \vb{x} : f(\vb{x})^{q-1} \neq 1$ is a witness of
$\exists \vb{x} : f(\vb{x}) = 0$.
A decision problem $A(e)$ can
be \emph{reduced in polynomial time} to $B(e)$ if there is a polynomial time
algorithm $P$ that transforms the input $e$ for $A$ into an input
$P(e)$ for $B$ such that $A(e)$ holds if and only if $B(P(e))$ holds.
We say that two decision problems are \emph{polytime-equivalent}
if each problem can
be reduced to the other one in polynomial time. For the group
$S_4$, we obtain the following connection between the various problems
associated to it:
\begin{thm} \label{thm:S4}
  Let $X = \{x_i \mid i \in \N_0\}$, let $\ob{F}_4$ be the field
  with four elements, let $S_4$ be the symmetric group on $4$ letters,
  and let $\alpha$ be a primitive element
  of $\ob{F}_4$.
  The following problems are polytime-equivalent: 
  \begin{enumerate}
  \item \label{it:ppoleqv} \emph{Polynomial nonequivalence over $S_4$}, $\textsc{coPolEqv} (S_4)$: Given a word $w \in (S_4 \cup X)^*$,
    decide whether there are an assignment $a:X \to S_4$ and
    $g \in S_4 \setminus \{1\}$ 
    such that $w$ evaluates to $g$ under the assignment $a$.
   \item \label{it:ppolsat} \emph{Polynomial satisfiability over $S_4$}, $\textsc{PolSat } (S_4)$: Given a word $w \in (S_4 \cup X)^*$,
    decide whether there is an assignment $a:X \to S_4$ 
    such that $w$ evaluates to $1$ under the assignment $a$.
  \item \label{it:parith} \emph{An arithmetic nonequivalence
  problem over $\ob{F}_4$}: Given $k, n \in \N$ and expanded polynomials $p_1, \ldots, p_k \in \Z_3 [x_1, \ldots, x_n]$, decide whether
    there exists $\vb{y} \in \{-1,1\}^n$ such that
    $\sum_{i=1}^k \alpha^{p_i(y_1, \ldots, y_n)} \neq 0$. 
  \item \label{it:pcircuits} \emph{Nonequivalence of $\textrm{CC}[2,3,2]$-circuits}:
    Given a $\textrm{CC}[2,3,2]$-circuit $C$, decide whether there is
    an input vector $x$ such that $C$ does not evaluate to~$1$ on input $x$.
      \item \label{it:prespol} \emph{Restricted polynomial nonequivalence over $\Mat_2 (\Z_2)$}: Given
    an expanded polynomial expression $f$ over the matrix ring
    $\Mat_2 (\Z_2)$ in the (noncommuting)
    variables $X$ such that all constants appearing in $f$ are invertible
    matrices, decide whether
    there is an assignment
    $a: X \to \Mat_2(\Z_2)$ such that all $a(x_i)$ are invertible matrices
    and $f$
    does not evaluate to $0$ under the assignment $a$
      (see also Definition~\ref{de:respol}).
  \end{enumerate}
\end{thm}
Writing $A \le_P B$ if $A$ can be reduced to $B$ in polynomial time,
we prove
\eqref{it:ppolsat} $\le_P$ \eqref{it:prespol} in
Theorem~\ref{thm:polsattorespoleqv},
\eqref{it:prespol} $\le_P$ \eqref{it:ppoleqv} in
Theorem~\ref{resPolEqvToPolEqv},
  \eqref{it:ppoleqv} $\le_P$ \eqref{it:ppolsat} in
  Theorem~\ref{thm:poleqvtopolsat},
  \eqref{it:prespol} $\le_P$ \eqref{it:parith} in
  Theorem~\ref{thm:restop1},
  \eqref{it:parith} $\le_P$ \eqref{it:prespol} in
  Theorem~\ref{thm:back2},
  \eqref{it:parith} $\le_P$ \eqref{it:pcircuits} in
  Theorem~\ref{thm:circ} and
  \eqref{it:pcircuits} $\le_P$ \eqref{it:parith} in
  Theorem~\ref{thm:back1}.

The reductions used in proving Theorem~\ref{thm:S4}
allow us to prove the following theorem, already conjectured in \cite{IKK:ESIS}:
\begin{thm}
\label{thm:algo}
Let $\gamma: \N \rightarrow \N$ be a monotone function such that every $\operatorname{CC}[2,3,2]$-circuit computing $\operatorname{AND}_n$ has size at least $\gamma(n).$
    Then there is $d>0$ such that $\textsc{PolSat}(S_4)$ can be solved in deterministic $\exp(O(\log(n)\gamma^{-1}(n^d)))$ time and probabilistic $\exp(O(\log(n) + \gamma^{-1}(n^d)))$ time.
\end{thm}
The proof will be given in~Section~\ref{sec:circuits}.
Assuming that SESH holds for $\operatorname{CC}[2,3,2]$-circuits, this yields an $\exp(c \log(n)^3)$ time algorithm for $\textsc{PolSat}(S_4).$ 

\section{Preliminaries on Groups and Matrices}
In this section, we prove results about groups and matrices that will be needed later. We write $\N$ for the set of positive integers, $\N_0 :=
\N \cup \{0\}$, and
for $k \in \N_0$, $\ul{k} = \{x \in \N \mid 1 \le x \le k\}$. 
\begin{lem}
\label{lem:shortProd}
Let $G$ be a group of order $k$, let $n\in\N$ and let $g_1,\dotsc,g_n, a\in G$ with $\prod_{i=1}^n g_i = a$. 
Then there exist an $m\in\N$ with $1 \leq m \leq k$ and indices $1 \leq i_1<\dots< i_m\leq n$ with $\prod_{j=1}^{m} g_{i_j} = a.$
\end{lem}
\begin{proof}
We prove the statement by induction on $n$.
If $n \leq k$, the claim obviously holds.
If $n > k$, define $b(l):=\prod_{i=1}^l g_i$ for $l \in \{1,\dots,n\}$. 
Since $b:\{1,\dots,n\}\rightarrow{G}$ maps $n$ elements to $k < n$ elements, we have $u<v\in\{1,\dots,n\}$ with $b(u)=b(v)$.
    Therefore $b(u)^{-1}b(v)=1=\prod_{i=u+1}^v g_i$, which implies $\prod_{i=1}^u g_i \prod_{i=v+1}^n g_i =a$.
    This product has less than $n$ factors, thus by the induction hypothesis, it has a subproduct of length at most $k$ as required.
\end{proof}

For $m \in \N$ and a prime power $q$, $\matrixRing$ denotes the set of
$m \times m$-matrices and $\matrixGroup$ the
set of invertible $m\times m$-matrices over $\baseField$.
\begin{lem}
\label{lem:sum}
Let $m\in\N$ and let $q$ be a prime power.
Let $A \in \matrixRing.$ 
Then either $m=1, q=2$ and $A = (1)$, or there exist $B,C \in \matrixGroup$ with $A = B+C.$
\end{lem}
\begin{proof}
We use linear algebra to obtain $P,Q \in \matrixGroup$
   with $A = P E_k Q$, where
   $E_k$ is the diagonal matrix whose first $k$ entries in the diagonal
   are $1$ and the other $m-k$ entries are $0$. 
   The identity matrix is denoted by $I$.
   If $q > 2$, then we choose an element
   $a \in \ob{F}_q \setminus \{0,1\}$ and write
   $E_k = a I + (E_k - a I)$. Both summands are diagonal matrices,
   and all their diagonal entries are in $\{a, 1-a,0-a \}$, hence
   nonzero, and therefore both matrices are invertible.
   Hence $P (aI) Q + P(E_k - a I)Q$ is the required representation.

   If $q = 2$, then we first observe that
   the $0$-matrix of every format can be written as $I + I$,
   where $I$ is the identity matrix.
   Furthermore,
   $\MatTwo{1}{0}{0}{0} =  \MatTwo{0}{1}{1}{0} +
   \MatTwo{1}{1}{1}{0}$,
   $\MatTwo{1}{0}{0}{1} =  \MatTwo{0}{1}{1}{1} +
   \MatTwo{1}{1}{1}{0}$, and
   $\MatThree{1}{0}{0}{0}{1}{0}{0}{0}{1} =
   \MatThree{1}{0}{1}{0}{0}{1}{1}{1}{1} +
   \MatThree{0}{0}{1}{0}{1}{1}{1}{1}{0}$.
    We will now see that the diagonal of $E_k$ can be divided into blocks
    consisting only of these matrices:
   If $k$ is even, we divide $E_k$ into $k/2$ blocks
   of $2 \times 2$-matrices and an $(m-k) \times (m-k)$ zero-matrix
   by writing
   \[
   E_k = \left( \begin{smallmatrix}
     \MatTwo{1}{0}{0}{1} &  &  & 0 \\[-2mm]
      & \ddots &  &  \\
      &  & \MatTwo{1}{0}{0}{1} &  \\[1mm]
     0 &  &  & \left(
     \begin{smallmatrix}
       {0} & {\cdots} & {0} \\[-1.5mm]
       {\vdots } & {\ddots} & {\vdots } \\
       {0} & {\cdots} & {0}.
     \end{smallmatrix}
      \right)
   \end{smallmatrix} \right).
   \]
     If $k$ is odd and $k \ge 3$, we divide $E_k$ into
     one block of size $3$, $\frac{k-3}{2}$ blocks of size $2$ and
     the zero matrix of size $(m-k) \times (m-k)$. 
     If $k = 1$ and $m=1$, then the first alternative in the statement of
     the Theorem occurs. If $k=1$ and $m \ge 2$,
     we divide $E_k$ into the $2 \times 2$-block
   $\MatTwo{1}{0}{0}{0}$ and the $(m-2) \times (m-2)$ zero matrix.
   Each of these blocks can be written as a sum of two invertible
   matrices, which yields the required decomposition of $E_k$, and
   hence of $A$. 
\end{proof}

We will now introduce some notation concerning polynomials over
a noncommutative ring $R$ with unity. For any algebraic structure,
a \emph{polynomial expression} is a well-formed term consisting
of variables, function symbols and constants from the algebra.
For a ring, we call a polynomial expression \emph{expanded}
if it is a sum of products. We call such an expanded polynomial
expression $p$ a \emph{restricted polynomial expression}
if all constants from $R$ appearing in $p$ are invertible
elements of $R$.
To state this more precisely,
let $X = \{x_1, \ldots, x_n\}$ be a set
of variables. A \emph{restricted monomial expression}
is a product of variables and invertible elements of $R$; for example,
when $r_1, r_2$ are \emph{invertible} elements from $R$, then
$x_2 r_1 x_2 x_1 x_1 r_2$ is a restricted polynomial expression.
Formally, we see a restricted monomial expression
as a nonempty word in $(X\cup U(R))^\ast$, where $U(R)$ is the set of
invertible elements of $R$. Hence for the matrix ring
$\matrixRing$, a restricted monomial expression is an
element of $(X \cup \matrixGroup)^{\ast}$.
A \emph{restricted polynomial expression} is a sum of
restricted monomial expressions; we will use $0$ to denote
the restricted polynomial expression which is the (empty) sum of zero monomial expressions.
We will now define the \emph{length} of such expressions.
For a word $e \in (X \cup U(R))^\ast$, the length
$|e|$ is the length of the word $e$, i.e., the $n \in \N_0$
such that $e \in (X \cup U(R))^n$. The length of
the restricted polynomial expression $p = \sum_{i=1}^k m_k$
is defined by $|p| := \sum_{i=1}^k |m_k|$.

 \begin{pro} \label{pro:Riscomplete}
   Let $m \in \N$, let $q$ be a prime power, and let $R$ be the
   matrix ring $\matrixRing$.
   Let $k \in \N$, and let $f$ be a function from $R^k$ to $R$.
   Then there exists a restricted polynomial expression
   $p(x_1, \ldots, x_k)$ that induces the function $f$ on $R$.
 \end{pro}
 \begin{proof}
   Since $R$ is a finite simple ring with
   unit, every $k$-ary function on $R$ is a polynomial function
   \cite{Kn:SVFC, KP:PCIA}. Therefore $f$ is induced
   by a polynomial expression $p_1$ over $R$ using the operation
   symbols $+,-,\cdot, 0, 1$. We will now
   transform $p_1$ into the required form. First of all,
   $p_1$ may contain constants from $R$ that are not invertible.
   Then 
   using Lemma \ref{lem:sum}, we replace every non invertible matrix
   in $p_1$ by a sum of two
   invertible matrices and obtain a polynomial expression $p_2$
   such that all constants occurring in $p_2$ are invertible elements
   of $R$. By expanding $p_2$, we obtain $p_3$ that is the sum
   of summands of the form $e$ and $-e$, where $e$ is a monomial.
   Replacing the constant $1$ by the
   identity matrix $I$,  $-e$ by $(-I) \cdot e$, 
   and omitting every monomial containing the operation symbol
   $0$, we obtain
   a restricted polynomial
   expression $p_4$ that induces the function $f$.
 \end{proof}

 \section{Reducing equation solvability over groups to an equivalence problem over matrix rings}
 In \cite{HS:TCOC}, Horv\'ath and Szab\'o 
 determined the complexity of polynomial equivalence
 testing over the group $S_3$. They used that
 $S_3$ is a semidirect product of $\Z_3$ with is automorphism group,
 which allowed them to translate polynomial equivalence over $S_3$ into a problem on polynomials over the field $\ob{F}_3$.
  The group $S_4$ is isomorphic to the semidirect product of the abelian group
 $\ob{Z}_2 \times \ob{Z}_2$ with its automorphism group
  $\operatorname{GL}_2 (\ob{F}_2)$; in other words, $S_4$ is isomorphic to
 the \emph{holomorph} of $\ob{Z}_2 \times \ob{Z}_2$.
  Proceeding similar to \cite{HS:TCOC},
  we will translate a question on $S_4$ into a question
  on polynomials over the noncommutative matrix ring $\matrixRingZZ$. 
  In contrast to \cite{HS:TCOC}, we will reduce the
  polynomial \emph{satisfiability} problem $\PolSat(S_4)$.
 In particular, Theorem~\ref{thm:polsattorespoleqv}
 reduces $\PolSat(S_4)$ to the \emph{restricted polynomial equivalence problem}
 over $\matrixRing$; the latter problem asks whether a restricted polynomial
 expression over $\matrixRingZZ$ evaluates to $0$ on all \emph{invertible}
 inputs; for example, $x_1x_1x_1x_1x_1x_1 + \MatTwo{1}{0}{0}{1}$ evaluates to $0$ for
 all $x_1$ in $\operatorname{GL}_2 (\ob{F}_2)$. 
 The group $S_4$ is contained in the class of all groups
 that are semidirect products of the additive group of the
 $m$-dimensional vector space over a $q$-element field
 $\vectorSpace$ with the automorphism group of this vector space
 $\matrixGroup$. 
 \begin{de}
Let $\mathbb{F}_q$ be the field of size $q$ and $m \in \N$. 
The \emph{holomorph} $\Holomorph$ is the group $\{\VecTwo{a}{B} \mid a \in \mathbb{F}_q^m, B \in \matrixGroup\}$
 with multiplication defined as $\VecTwo{v}{A}\cdot\VecTwo{w}{B}:=\VecTwo{v + Aw}{AB}$.
 \end{de}
 In this setting, $S_4$ is isomorphic to $\operatorname{Hol}(2,2)$.
 We note that if $q$ is not prime then the name \emph{holomorph}
 is usually used for the semidirect product of
 $\vectorSpace$ with its full automorphism group. If $q = p^k$,
 then the automorphism group of $\vectorSpace$ is
 $\mathrm{GL}_{mk} (\Z_p)$; our usage here is different,
 since we defined $\Holomorph$ as the semidirect product
 of $\vectorSpace$ with $\matrixGroup$.
 Our results bring new insights mainly for the cases
 $(m, q) \in \{ (2,2), (2,3) \}$: for $m = 1$,
 $\Holomorph$ is meta-abelian, and for
 such a group $G$,  $\PolSat(G)$ and
 $\PolEqv (G)$ lie in $\PP$ \cite[Corollary 8]{Ho:TCOT}
 and for $m \ge 2$ and
 $(m,q) \not\in \{ (2,2), (2,3) \}$, the group
 $\Holomorph$ is not solvable, and thus
 $\PolSat (G)$ is $\NP$-complete \cite{GR:TCOS} and  
 $\PolEqv (G)$ is co-$\NP$-complete by \cite{HLMS:TCOT}.
 The group $\operatorname{Hol}(2,3)$ is a 
 solvable group of order $432$. 
 
 In these holomorphs, the product of $n$ elements can be computed using
 the formula
 \begin{equation}  \label{eq:product}
 \prod_{i=1}^n \VecTwo{v_i}{A_i} = \VecTwo{ \sum_{i=1}^n (\prod_{j=1}^{i-1} A_j) v_i}{\prod_{i=1}^n A_i},
 \end{equation}
 which can easily be proved by induction.
 
\begin{de} \label{de:respol}
  Given a restricted polynomial expression $p$, the \emph{restricted polynomial equivalence problem} asks whether $p$ evaluates to $0$ on all invertible inputs, that is, whether $p(\vb{x})=0$ for all  $\vb{x}\in \matrixGroup^n$.
\end{de}

In the following two lemmata, we bound the length of the product and the functional composition
of two restricted polynomial expressions.
\begin{lem}
\label{productLemma}
Let $p_1,\dots,p_k$ be restricted polynomial expressions,
and let $q$ be the expansion of $\prod_{i=1}^k p_i$. Then we have 
$\length{
q
} \leq (\sum_{i=1}^k \length{p_i})(\prod_{i=1}^k \length{p_i})$
\end{lem}
\begin{proof}
Since each $p_i$ has at most $\length{p_i}$ restricted monomial expressions, there can be only $\prod_{i=1}^k \length{p_i}$ restricted monomial expressions in the resulting product. 
Since a restricted monomial expression in $p_i$ can have length at most $\length{p_i}$, any product of such restricted monomial expressions can have length at most $\sum_{i=1}^k \length{p_i}$.
\end{proof}

\begin{lem}
\label{compositionLemma}
Let $p$ and $q$ be restricted polynomial expressions, where $q$ contains only one variable $x$ and
$p \neq 0$. 
The restricted polynomial expression  $q\circ p$ obtained from $q$ by replacing all occurrences of $x$ by
$p$ and fully expanding the result then satisfies
 $\length{q\circ p} \leq \length{q}\cdot\length{p}^{\length{q}}.$
\end{lem}
\begin{proof}
Let $n$ be the number of monomial expressions in $q$.
We proceed by induction on $n$. If $n = 0$, then $q = q \circ p = 0$, and thus
$|q \circ p| = |q| = 0$, which implies that the claimed inequality holds.
If $n = 1,$ then $q$ is a restricted monomial expression
$q_1 x q_2 \cdots q_r x q_{r+1}$, where $q_1, \ldots, q_{r+1}$ are (possibly empty)
words over the alphabet $\matrixGroup$, and $p = \sum_{i = 1}^l m_i$, where all
$m_i$ are restricted monomials and $l \ge 1$. Then
\[
q \circ p =  \sum_{i: \ul{r} \to \ul{l}} q_1 m_{i(1)} q_2 \cdots q_r m_{i(r)} q_{r+1},
\]
and therefore
\[
  \begin{split}
    |q \circ p| & = \sum_{i: \ul{r} \to \ul{l}} \big( \sum_{j=1}^{r+1} |q_j|  + \sum_{j=1}^r |m_{i(j)}| \big) 
     = l^r \sum_{j=1}^{r+1} |q_j| +
    \sum_{j=1}^r \sum_{i:\ul{r} \to \ul{l}} |m_{i(j)}| \\  
    & = l^r \sum_{j=1}^{r+1} |q_j| + r l^{r-1} \sum_{k=1}^l |m_k| 
    \le |p|^r \big(\sum_{j=1}^{r+1} |q_j| \big) + r \cdot |p|^{r-1} \cdot |p| \\
    & = |p|^r \big( (\sum_{j=1}^{r+1} |q_j|) + r \big) 
     = |p|^r |q|  \le |p|^{|q|} |q|.
  \end{split}
\]  
For the induction step, we assume that $n > 1$ and that $q=q_1+q'$ for a restricted monomial expression $q_1$ and a restricted polynomial expression $q'$ with $n-1$ monomials.
If $|p| = 1$, then $|q \circ p| = |q| = |q| 1^{|q|} = |q| |p|^{|q|}$, and hence the claimed inequality holds.
In the case $|p| > 1$, we have $\length{q\circ p} = 
\length{q_1\circ p}+\length{q'\circ p} \leq
\length{q_1}\cdot \length{p}^{\length{q_1}} + \length{q'}\cdot \length{p}^{\length{q'}}
\leq
\length{q}\cdot \length{p}^{\length{q_1}} + \length{q}\cdot \length{p}^{\length{q'}}
= \length{q} \cdot (\length{p}^{\length{q_1}} + \length{p}^{\length{q'}})$.
Since $\length{p} > 1$, $\length{q_1} > 0$ and $\length{q'} > 0$, and since
$a+b \le ab$ for $a,b \ge 2$, we have
$\length{q} \cdot (\length{p}^{\length{q_1}} + \length{p}^{\length{q'}}) \le
\length{q} \cdot \length{p}^{\length{q_1}} \cdot \length{p}^{\length{q'}} =
 \length{q}\cdot \length{p}^{\length{q_1}+\length{q'}}
= \length{q} \cdot  \length{p}^{\length{q}}$.
 \end{proof}

Over a finite field $\baseField$, one often uses that
the solutions of $p_1 = p_2 = 0$ are those values for which
$(1-p_1^{q-1})(1-p_2^{q-1}) \neq 0$. A similar construction is
also possible over $\matrixRing$. 

\begin{lem} \label{ResPolSatToResPolEqv}
    Let $m\in\N$ and let $q$ be a prime power with $(m,q) \neq (1,2)$.
    Then there is $C > 0$ and a polynomial time algorithm that takes a pair $(p_1, p_2)$ of restricted polynomial expressions in the variables $x = (x_1,\dots,x_n)$ and returns a restricted polynomial expression $p_3$ with the following properties: 
\begin{enumerate}
\item \label{it:p1} For all $y \in \matrixRing^n \, \colon \,
      (p_1(y) = 0 \text{ and } p_2(y) = 0) \Longleftrightarrow p_3(y) \neq 0$. 
    \item \label{it:p2} The length of $p_3$ satisfies  $|p_3| \le 2C^3(\length{p_1}\cdot \length{p_2})^{2C}.$
\end{enumerate}
\end{lem}
\begin{proof}
  We use Proposition \ref{pro:Riscomplete} to obtain a restricted polynomial expression $f$ with the property that $f(0) = I$ and $f(y) = 0$ for $y \in \matrixRing \setminus \{0\}$, and we define a polynomial expression
  $p_3$ as the expanded form of  $f(p_1(x))\cdot f(p_2(x))$; then   
  \[
    p_3 (y) =  f(p_1(y))\cdot f(p_2(y))
    \]
  for all $y \in \matrixRing^n$.  
Hence $p_3(x)\neq 0$ is equivalent to $f(p_1(x)) = f(p_2(x)) = I$, which is equivalent to $p_1(x) = p_2(x) = 0$,
which proves item~\eqref{it:p1}.
For proving~\eqref{it:p2}, we set $C:=\length{f}.$
    By Lemma \ref{compositionLemma} we have $\length{f(p_i(x))} \leq C\length{p_i}^C$ for $i \in \{1,2\}$. 
    Therefore by Lemma \ref{productLemma} we have
    $\length{p_3} \leq (C\length{p_1}^C + C\length{p_2}^C) C^2\length{p_1}^C\length{p_2}^C =
    C (\length{p_1}^C + \length{p_2}^C) C^2\length{p_1}^C\length{p_2}^C$.
    Since for $a,b \ge 1$ we have $a+b \le 2ab$, the last expression is at most
    $2 C^3 (|p_1| \cdot |p_2|)^{2C}$. 
\end{proof}
We are now ready to start the reduction of $\PolSat (\Holomorph)$.
In a first step, we reduce $\PolSat(\Holomorph)$ to a disjunction
of restricted polynomial inequalities.
\begin{lem}
    \label{HolSatToResPolEqv}
    Let $m\in\N$, let $q$ be a prime power with $(m,q) \neq (1,2)$, and
    let $l:= \# \vectorSpace = q^m$.
    Then there is a $C > 0$ and a polynomial time algorithm that takes a group word $p$ over $\Holomorph$ in
    $n$ variables of length $k$ and returns a set $E$ of restricted polynomial expressions with the following properties:
    \begin{enumerate}
        \item \label{it:e1}  The cardinality of $E$ is at most $(nl)^l$.
        \item \label{it:e2} For every $e\in E$, the length of $e$ satisfies $|e| \le 2C^3(2k^2(k+1))^{2C}$.
        \item \label{it:e3} There exists $x\in \Holomorph^n$ with $p(x)=1$ if and only if there exist $e\in E$ and $Z \in \matrixGroup^n$ with $e(Z) \neq 0$.
    \end{enumerate}
\end{lem}
\begin{proof}
  We assume that $p = w_1 \dots w_k$ is a word in $(\Holomorph \cup \{x_1, \dots, x_n\})^*$ of length $k>0$.
  We will now use that $\Holomorph$ is a semidirect product, and therefore we split each variable
  $x_j$ (which is intended to range over $\Holomorph$) in the word $p$ into two variables $y_j$ and $Z_j$, where
  $y_j$ ranges over $\vectorSpace$ and $Z_j$ ranges over $\matrixGroup$.
  In this way, we obtain a word $$p' = \VecTwo{u_1}{V_1} \dots \VecTwo{u_k}{V_k}.$$ Formally,
  if the $i$\,th letter  $w_i$ of $p$ is a variable $x_j$, then $\VecTwo{u_i}{V_i} := \VecTwo{y_j}{Z_j}$, and if
  $w_i \in \Holomorph$, then $u_i$ and $V_i$ are the components of $w_i$, i.e.,
  $\VecTwo{u_i}{V_i} := w_i$. We now use~\eqref{eq:product} to obtain the  expression
    $$
    p '' = \VecTwo{I u_1 + V_1u_2 + \dots + V_1\cdots V_{k-1}u_k}{V_1\cdots V_k}, 
    $$
     where $I$ is the $m \times m$ identity matrix. Then
    for each assignment
    to the variables $y_1, Z_1, \ldots, y_n, Z_n$, the expressions $p'$ and $p''$
    evaluate to the same element of $\Holomorph$. 
    
    Now let $c_1, \dots, c_l$ be an enumeration of the elements of $\vectorSpace$.
    Using the \emph{collecting procedure} introduced in \cite{HS:TCOC}, we group the restricted monomials
    $I, V_1, \ldots, V_1 \cdots V_{k-1}$ 
    in the first component of $p''$ according to the $y_i$ or $c_i$ they act on. Additionally setting $p_0 := V_1\cdots V_k$,
    we obtain restricted polynomial expressions $p_0, \dots, p_{n+l}$ such that 
    $$
     \VecTwo{p_1(Z)y_1 + \dots + p_n(Z)y_n + p_{n+1}(Z)c_1 + \dots + p_{n+l}(Z)c_l}{p_0(Z)}
    $$
     is equivalent to $p'$ and $p''$. 

    The equation $p(x) = 1$ now has a solution $x\in \Holomorph^n$ if and only if the system 
    \begin{equation} \label{eq:collected}
      p_1(Z)y_1 + \dots + p_n(Z)y_n \, + \, p_{n+1}(Z)c_1 + \dots + p_{n+l}(Z)c_l = 0 \,\, \land \,\, p_0(Z) = 1
    \end{equation}
    has a
    solution $Z\in (\matrixGroup)^n, y\in (\vectorSpace)^n$. Suppose now that $(y,Z)$
    solves~\eqref{eq:collected}.
    Using Lemma \ref{lem:shortProd}, there are $d \in \N_0$ with $d \le l$ and
    $1 \le i_1 < \cdots < i_{d} \le n$ such that
    \[
    p_1(Z)y_1 + \cdots + p_n(Z)y_n = p_{i_1}(Z) y_{i_1} + \cdots p_{i_{d}} (Z) y_{i_{d}}.
    \]
    Then we obtain a new solution $(y', Z')$ of~\eqref{eq:collected} by setting
    $y'_j := 0$ for $j \not\in \{i_1, \ldots, i_{d}\}$, $y'_{i_j} := y_{i_j}$ for $j \in \{1,\ldots, d\}$ and
    $Z' := Z$. Therefore, when looking for a solution of~\eqref{eq:collected}, it is sufficient
    to look for a solution in which 
    $(y_1, \dots, y_n)$ is an element of
    \[ T := \{v \in (\vectorSpace)^n \mid \#\{i: v_i\neq 0\} \leq l\}. \]
    Now since $l$ depends only on the group, but not on the input $p$,
    the size of $T$ is small enough for our purposes (polynomial in $n$) so that we can look for
    solutions $(y,Z)$ of~\eqref{eq:collected} by trying to find a solution for each setting $y := v \in T$
    separately.
    Hence we fix $v \in T$, and we are going to define a restricted polynomial expression $e_v$ with the following property:
    \begin{multline} \label{eq:pe1}
      \text{For all } Z \in \matrixGroup^n \, \colon \, e_v(Z) \neq 0
      \Longleftrightarrow \\
      \big(p_1(Z)v_1 + \dots + p_n(Z)v_n \, + \,
       p_{n+1}(Z)c_1 + \dots + p_{n+l}(Z)c_l = 0 \, \land \,  p_0(Z) = 1\big).
    \end{multline}
    To this end,
    let $i_1<\dots<i_d$ be the indices where $v$ is nonzero. Since we wish to obtain an equation over the matrix ring
    $\matrixRing$, we replace a vector $v_i \in \vectorSpace$ appearing in~\eqref{eq:pe1} by a matrix
    $W  = (v_i \,\, 0)$ whose
    first column is $v_i$ and whose other columns are $0$; this matrix $W$ will not be regular (unless $m = 1$ and $v_i \neq 0$),
    but Lemma~\ref{lem:sum} allows us to write $W$ as the sum of two regular matrices $M$ and $N$;
    similarly, we replace the $c_i$. Stated more
    precisely,
    Lemma \ref{lem:sum} allows us to  find $M_1, \dots, M_{d+l}, N_1, \dots, N_{d+l} \in \matrixGroup$ such that
    $M_j + N_j = (v_{i_j} \; 0)$ for all $j \in \{1,\ldots,d\}$ and $M_j + N_j = (c_{j-d}\; 0)$ for all
    $j \in \{d+1, \ldots,  d+l\}$. We define
    \begin{multline*}
    a_v(Z) :=p_{i_1}(Z)(M_1 + N_1) + \dots + p_{i_d}(Z)(M_d + N_d) \\ + p_{n+1}(Z)(M_{d+1} + N_{d+1})+\dots + p_{n+l}(Z)(M_{d+l} + N_{d+l}).
    \end{multline*}
    Now for each $Z \in \matrixRing$,
    $$p_1(Z)v_1 + \dots + p_n(Z)v_n + p_{n+1}(Z)c_1 + \dots + p_{n+l}(Z)c_l = 0$$
    if and only if $a_v(Z) = 0$. 
    By expanding $a_v (Z)$, we obtain a restricted polynomial expression $b_v(Z)$ with
    $a_v(Z) =b_v (Z)$ for all $Z \in \matrixRing$.
    Now we invoke the algorithm from Lemma \ref{ResPolSatToResPolEqv} on the pair $(b_v(Z), p_0(Z) + (-I))$ to obtain a restricted polynomial expression $e_v (Z)$ that satisfies~\eqref{eq:pe1}.
    We claim that the set
    \[
    E := \{ e_v(Z) \mid v \in T \}
    \]
    satisfies all requirements in the statement of the lemma.

    For item~\eqref{it:e1}, we observe that for $n \ge l$, the set $T$ has at most
    $\VecTwo{n}{l} l^l$ elements (choose $l$ possibly nonzero entries first, and then choose
    one of $l$ values for each of these $l$ entries), and thus $|T| \le n^l l^l = (nl)^l$.
    If $n < l$, then $|T| = l^n$, and hence $|T| \le l^l \le (nl)^l$.
    
    For item~\eqref{it:e2}, we observe that $p_1(Z), \ldots, p_n (Z)$ together contain
    at most $k$ monomials because each of these monomials comes from $p''$.
    Since each of these monomials is of length at most $k$, we obtain that
    $b_v (Z)$ is of length at most $2k^2$.
    The polynomial $p_0 (Z)$ is of length $k$, hence $p_0 (Z) + (-I)$ is of length $k+1$.
    From Lemma \ref{ResPolSatToResPolEqv}, we obtain
    $|e_v| \le 2C^3 (2k^2 (k+1))^{2C}$, where $C$ is the constant from this Lemma.

    For item~\eqref{it:e3}, we observe that $p(x) = 1$ is equivalent to the system~\eqref{eq:collected},
    which has a solution in $\Holomorph$ if and only if  there are $v \in T$ and $Z \in \matrixGroup^n$ such
    that $e_v (Z) \neq 0$.
\end{proof}

We note that the disjunction
of the inequalities $p_i (x) \neq 0$ is satisfiable if and only if
$\sum p_i (x) y_i \neq 0$ is satisfiable, where the $y_i$'s are new
variables. The next lemma is a similar construction that takes into
account that we wish to set the $y_i$'s only to invertible matrices
in $\matrixRing$.
\begin{lem}
\label{lem:manyeqv}
Let $m\in\N$ and let $q$ be a prime power with $(m,q) \neq (1,2)$.
There is an algorithm that takes restricted polynomial expressions $p_1, \dots, p_k$ over $\matrixRing$ and returns 
a restricted polynomial expression $u$ of size at most $2k + 2(\length{p_1} + \dots + \length{p_k})$ such that 
$$
\exists i \in \ul{k} \,\, \exists \vb{x}\in \matrixGroup^n \,:\, p_i(\vb{x}) \neq 0
$$
if and only if
$$
\exists \vb{x} \in \matrixGroup^{n}, \vb{y}\in \matrixGroup^{2k}: u(\vb{x},\vb{y}) \neq 0.
$$
This algorithm runs in time polynomial in $\length{u}$.
\end{lem}
\begin{proof}
We let $u$ be the expansion of 
\[
    \sum_{i=1}^k (y_i^{(1)} + y_i^{(2)})p_i(x).
\]
For the ``if''-direction, assume that we have $\vb{x} \in \matrixGroup^{n}, \vb{y}\in \matrixGroup^{2k}$ such that $u(\vb{x},\vb{y}) \neq 0.$
Then there must be some $i\in\{1,\dots, k\}$ such that $p_i(\vb{x}) \neq 0.$
For the ``only if''-direction, assume that we have $\vb{x}\in \matrixGroup^n$ and $ i \in \ul{k}$ such that $p_i(\vb{x}) \neq 0$.
We use Lemma \ref{lem:sum} to set $y_i^{(1)}$ and $y_i^{(2)}$ such that $y_i^{(1)} + y_i^{(2)} = I$. We set $(y_j^{(1)}, y_j^{(2)}) := (I, -I)$ for $j\neq i.$
The assignment $(\vb{x},\vb{y})$ then satisfies $u(\vb{x},\vb{y}) = p_i(\vb{x}) \neq 0.$
\end{proof}

The results in this section provide a polynomial time reduction
of
$\textsc{PolSat} (\Holomorph)$ to the polynomial equivalence problem over $\matrixRing$:
\begin{thm}
    \label{thm:polsattorespoleqv}
Let $m\in\N$ and let $q$ be a prime power with $(m,q) \neq (1,2)$.
 Then
 $\textsc{PolSat}(\Holomorph)$ can be reduced in polynomial time
 to the complement of the restricted polynomial equivalence problem over $\matrixRing.$
\end{thm}
\begin{proof}
    Let $w(x_1,\dots,x_n)$ be a group word of length $k$ over $\Holomorph$.
    By Lemma \ref{HolSatToResPolEqv}, we can construct a set $E$ such that $w(x_1, \ldots, x_n) = 1$ has a solution if and only there are
    $e \in E$ and $Z \in \matrixGroup^n$ with $e(Z) \neq 0$.
    This construction can be done in time polynomial in the length of $w$.
    Let $r:= |E|$.
    Using Lemma \ref{lem:manyeqv}, we can combine the elements of $E$ into a single restricted polynomial expression $p (x_1, \ldots, x_n, y_1, \ldots, y_{2r})$ such that
    there is $e \in E$ and  $Z \in \matrixGroup^{n}$ with
    $e (Z) \neq 0$ if and only if
    there is $(Z,U) \in \matrixGroup^{n+2r}$ with $p(Z, U) \neq 0$.
    Again, $p$ can be constructed in time polynomial in the length
    of $w$.
    Now we have reduced the question $\exists \vb{x} :w( \vb{x}) = 1$
    to the question whether
    $\forall (\vb{z}, \vb{u}) \in \matrixGroup^{n+2r} : p(\vb{z}, \vb{u}) = 0$
    is false, which is an instance of
    the complement of
    the restricted polynomial equivalence problem over $\matrixRing$.
\end{proof}

\section{Reducing an equivalence problems over matrix rings to
  polynomial equivalence over groups}

In this section, we provide a reduction from the restricted equivalence problem
over $\matrixRing$ to $\PolEqv (\Holomorph)$.

\begin{pro} \label{pro:idem}
  Let $m \in \N$ and let $q$ be a prime power with $(m,q) \neq (1,2)$.  Let
  $G := \operatorname{Hol}(m, q)$, and let $V := \{ \VecTwo{v}{I} \mid v \in \ob{F}_q^m \}$,
  where $I$ denotes the $m \times m$-identity matrix.
  We assume that $(m,q) \neq (1,2)$.
  Then there is a word $e(x) \in (G \cup \{x\})^*$ such
  that
  $e(g) \in V$ and $e(h) = h$ for all $g \in G$ and $h \in V$.
\end{pro}
\begin{proof}
   The group $V$ is a minimal normal subgroup of $G$, and its centralizer
   $C_G (V)$ is equal to $V$. Hence $V$ is the unique minimal normal subgroup
   of $G$. Taking $R$ to be the inner automorphism near-ring $I(G)$, \cite[Theorem~4.1(2)]{FK:UPOC}
   yields
   an $e \in R$ with $e(h) = h$ and $e(g) = 0$ for all
   $h \in V$ and $g \in G \setminus V$.
   Since $G$ is a finite group, this $e$ is induced by a word as in the
   statement of the Theorem.
\end{proof}
   We observe that  a suitable $e$  was also constructed in \cite[Proposition~7.14]{AI:PIIE}.

\begin{thm}  \label{resPolEqvToPolEqv}
  Let $m \in \N$ and let $q$ be a prime power with $(m,q) \neq (1,2)$.
  Then the restricted polynomial equivalence problem over
  $\matrixRing$ is polytime-reducible to
  $\PolEqv (\Holomorph)$.
\end{thm}
\begin{proof}
  For this reduction, we provide
  an algorithm that takes a restricted polynomial
  expression $p (x_1, \ldots, x_n)$
over $\matrixRing$ and returns
a group word $w(x_1, \ldots, x_{n+1}) \in   (\Holomorph \cup \{x_1, \ldots, x_{n+1}\})^*$ of length linear in $\length{p}$ such that 
$\forall \vb{g}\in \matrixGroup^n : p (\vb{g}) = 0$ if and only if
$\forall \vb{y}\in \Holomorph^{n+1}: w(\vb{y}) =  1$.
Moreover, this algorithm has running time polynomial in $\length{p}$.

  For a restricted monomial
  $\mu \in (\matrixGroup \cup \{x_1,\ldots,x_{n}\})^*$,
  we construct a word
  $\phi (\mu) \in (\Holomorph \cup \{x_1, \ldots, x_n\})^*$
  by replacing every letter $g \in \matrixGroup$ by
  the element $\VecTwo{0}{g}$ of $\Holomorph$.
  We claim that for all $u \in \vectorSpace$, for all
  $g_1, \ldots, g_n \in \matrixGroup$ and for all
  $v_1, \ldots, v_n \in \vectorSpace$, we have
  \begin{equation} \label{eq:phi}
    \VecTwo{\mu (g_1, \ldots, g_n) \cdot u}{I} =
    \phi (\mu) (\VecTwo{v_1}{g_1}, \ldots, \VecTwo{v_n}{g_n})
     \ast \VecTwo{u}{I} \ast 
     (\phi (\mu) (\VecTwo{v_1}{g_1}, \ldots, \VecTwo{v_n}{g_n}))^{-1},
  \end{equation}
  where $\ast$ is the multiplication of $\Holomorph$.
  We observe that from the definition of the multiplication in
  the holomorph, we obtain
  $\VecTwo{v}{g}^{-1} = \VecTwo{-g^{-1} \cdot v}{g^{-1}}$ and
  therefore
  \begin{multline} \label{eq:conjugation}
  \VecTwo{v}{g} \ast \VecTwo{u}{I} \ast \VecTwo{v}{g}^{-1}
   =
   \VecTwo{v}{g} \ast \VecTwo{u}{I} \ast \VecTwo{-g^{-1} \cdot v}{g^{-1}}
   \\ =
   \VecTwo{v+ g\cdot u}{g} \ast \VecTwo{-g^{-1} \cdot v}{g^{-1}}
   =
   \VecTwo{v + g \cdot u - v}{I} = \VecTwo{g \cdot u}{I}.
  \end{multline}
  Now in order to prove~\eqref{eq:phi}, we proceed by induction on
  the length of $\mu$. If $\mu = g$, then~\eqref{eq:conjugation}
  yields
  $\VecTwo{g \cdot u}{I} = \VecTwo{0}{g} \ast \VecTwo{u}{I} \ast
  \VecTwo{0}{g}^{-1} = \phi(g) \ast \VecTwo{u}{I} \ast \phi(g)^{-1}$.
  If $\mu = x_i$, then
  $\VecTwo{g_i \cdot u}{I} = \VecTwo{v_i}{g_i} \ast \VecTwo{u}{I} \ast
  \VecTwo{v_i}{g_i}^{-1} = \phi(\mu) (g_1, \ldots, g_n)
  \ast \VecTwo{u}{I} \ast (\phi(\mu) (g_1, \ldots, g_n))^{-1}$.
  For the induction step, we let $\mu = \mu_1 \ldots \mu_k$
  let $\alpha := \mu_1 (g_1, \ldots, g_n)$,
  $\beta := \mu_2 \ldots \mu_k (g_1, \ldots, g_n)$,
  $\phi_\alpha := \phi (\mu_1) (\VecTwo{v_1}{g_1}, \ldots,
  \VecTwo{v_n}{g_n})$ and
  $\phi_\beta := (\phi (\mu_2 \ldots \mu_k) (\VecTwo{v_1}{g_1}, \ldots,
  \VecTwo{v_n}{g_n})$.
  Then from what we have proved for words of length~$1$ and the
  induction hypothesis, we obtain
   \begin{multline*}
  \VecTwo{(\alpha \beta) \cdot u}{I} =
  \VecTwo{\alpha \cdot (\beta  \cdot u)}{I} =
  \phi_{\alpha} \ast \VecTwo{\beta \cdot u}{I} \ast \phi_{\alpha}^{-1} =
  \phi_{\alpha} \phi_{\beta} \VecTwo{u}{I} \phi_{\beta}^{-1} \phi_\alpha^{-1}
  \\ =
  \phi (\mu) (\VecTwo{v_1}{g_1}, \ldots, \VecTwo{v_n}{g_n})
     \ast \VecTwo{u}{I} \ast 
     (\phi (\mu) (\VecTwo{v_1}{g_1}, \ldots, \VecTwo{v_n}{g_n}))^{-1}.
   \end{multline*}
  This completes the proof of~\eqref{eq:phi}.

  We will now produce the required word $w(x_1,\ldots,x_{n+1})$
  from $p(x_1, \ldots, x_n)$.
  Let $V := \{ \VecTwo{u}{I} \mid u \in \vectorSpace \}$.
  By Proposition \ref{pro:idem}, there exists a
  word $e \in (\Holomorph \cup \{x\})^*$ such that
  the unary mapping induced by $e$ is idempotent with
  range $V$.
  
  We write $p = \sum_{i=1}^l \mu_i$ as sum of monomials and
  define $w$ by
  \[
  w (x_1, \ldots, x_{n+1}) =
  \prod_{i=1}^l \phi (\mu_i) e (x_{n+1}) (\phi (\mu_i))^{-1}.
  \]
  Here, for a word
  $\gamma \in (\Holomorph \cup \{x_1, \ldots, x_n\})^*$,
  the inverse $\gamma^{-1}$ is defined by
  $(x_i)^{-1} := (x_i^{|\Holomorph| - 1})$, $(h)^{-1} := (h^{-1})$ for
  words of length one and then recursively by
  $(\beta \delta)^{-1} = \delta^{-1} \beta^{-1}$. 
  
  The length of $w$ can then be bounded as follows:
  \[
     \begin{split}
       \length{w} & \le \sum_{i=1}^l (\length{e} +
       (1 + |\Holomorph| - 1) \length{\mu_i})
       \\
       & = l \length{e} + |\Holomorph| |p| \le
       (\length{e} + |\Holomorph|) |p|.
     \end{split}
  \]
  Now for all $u \in \vectorSpace$, we have
  \begin{equation} \label{eq:wp}
  \begin{split}
    w (\VecTwo{v_1}{g_1}, \ldots, \VecTwo{v_n}{g_n},
    \VecTwo{u}{I}) & =
    \prod_{i=1}^l \phi (\mu_i) (\VecTwo{\vb{v}}{\vb{g}})
    \VecTwo{u}{I} (\phi (\mu_i) (\VecTwo{\vb{v}}{\vb{g}}))^{-1}
    \\
    &=
    \prod_{i=1}^l \VecTwo{\mu_i (\vb{g}) \cdot u}{I} \\
    &=
    \VecTwo{\sum_{i=1}^l \mu_i (\vb{g}) \cdot u}{I} \\
    &=
     \VecTwo{p (\vb{g}) \cdot u}{I}. 
  \end{split}
  \end{equation}
  Now assume that $p(\vb{g}) = 0$ for all
  $\vb{g} \in \matrixGroup^n$, and let
  $\VecTwo{v_1}{h_1}, \ldots, \VecTwo{v_{n+1}}{h_{n+1}}
  \in \Holomorph$.
  Let $\VecTwo{u}{I} := e (\VecTwo{v_{n+1}}{h_{n+1}})$.
  Since $e$ is idempotent, we then have
  $w (\VecTwo{v_1}{g_1}, \ldots, \VecTwo{v_n}{g_n},
          \VecTwo{v_{n+1}}{g_{n+1}}) =
   w (\VecTwo{v_1}{g_1}, \ldots, \VecTwo{v_n}{g_n},\VecTwo{u}{I})
   =       
   \VecTwo{p (\vb{g}) \cdot u}{I} = \VecTwo{0}{I} = 1$.

   Now assume that there is $\vb{g} \in \matrixGroup^n$
   with $p(\vb{g}) \neq 0$. Then
     there is $u \in \vectorSpace$ with
     $p(\vb{g}) \cdot u \neq 0$,  and hence
     by~\eqref{eq:wp}, $w(\vb{g}, \VecTwo{u}{I}) \neq
     \VecTwo{0}{I}$.
 \end{proof}

\section{Reduction of polynomial equivalence to polynomial
         satisfiability for holomorphs}
In this section, we will see that $\textsc{PolSat} (S_4)$ and
$\textsc{PolEqv} (S_4)$ are -- in some sense -- of equal computational
complexity.
By  Theorems~\ref{thm:polsattorespoleqv} and~\ref{resPolEqvToPolEqv},
$\textsc{PolSat} (S_4)$ can be reduced to the complement of  $\textsc{PolEqv} (S_4)$.
For the converse reduction, 
we note that for any group, there is an obvious Turing reduction of
$\PolEqv (G)$ to $\PolSat(G)$ : 
$p = 1$ is valid if for all $g \in G \setminus \{1\}$,
$p = g$ is not satisfiable. Hence for every group $G$,
$\PolSat (G) \in \PP$ implies $\PolEqv (G) \in \PP$.
If $G = \Holomorph$, we even have a polytime reduction:
\begin{thm} \label{thm:poleqvtopolsat}
  Let $m \in \N$, let $q$ be a prime power with $(m,q) \neq (1,2)$, and let
  $G := \Holomorph$. Then
  $\textsc{coPolEqv} (G)$ is polytime-reducible to
  $\PolSat (G)$.
\end{thm}
\begin{proof}
  Let $H := \{\VecTwo{v}{I} \mid v \in \vectorSpace\}$, and
  let $a \in H \setminus \{ \VecTwo{0}{I} \}$.
  Then $H$ is the unique minimal normal subgroup of $G$,
  and its centralizer $C_G(H)$ is equal to $H$. 
  Thus $H$ is a homogeneous ideal of the group $G$ in the sense
  of \cite[Proposition~7.16]{AI:PIIE}. Now this Proposition yields
  a unary polynomial function $f$ of $G$ such that
  $f|_H = \mathrm{id}_H$ and $f(x) = a$ for $x \in G \setminus H$.
  There is a group word $e \in (G \cup \{x\})^*$ that induces the function
  $f$. The reduction is now as follows: in order to check whether for
  a given word $w \in (G \cup \{x_1, \ldots, x_n\})^*$, the property
  $\neg \forall \vb{x} : w (\vb{x}) = \VecTwo{0}{I}$ holds,
  we compute a new word $x_{n+1} \, e(w  (\vb{x})) \, x_{n+1}^{|\Holomorph|-1}$;
  the length of this word is bounded linearly in the length of $w$.
  Using a procedure for $\PolSat (G)$, we
  check whether
  \begin{equation} \label{eq:wex}
    x_{n+1} \, e(w (\vb{x})) \,  x_{n+1}^{|\Holomorph|-1} = a
  \end{equation}  
  is solvable.
  Clearly, if $x_{n+1} \, e(w (\vb{x})) \, x_{n+1}^{-1} = a$, then
  $e (w (\vb{x})) \neq \VecTwo{0}{I}$, and thus $w(\vb{x}) \neq
  \VecTwo{0}{I}$. For the other direction, we
  assume that $w (\vb{x}) \neq \VecTwo{0}{I}$. If $w (\vb{x}) \not\in H$,
  then $e (w (\vb{x})) = a$, and so setting $x_{n+1} := \VecTwo{0}{I}$
  yields a solution of~\eqref{eq:wex}. If $w (\vb{x}) \in H$, then
  there is $A \in \matrixGroup$ with
  $\VecTwo{0}{A} w (\vb{x}) \VecTwo{0}{A}^{-1}  = a$, and so
  setting $x_{n+1} := \VecTwo{0}{A}$, we obtain a solution of~\eqref{eq:wex}.
  \end{proof}

\section{An arithmetical problem over finite fields equivalent to
  $\PolEqv(S_4)$}
It follows from Theorems~\ref{thm:poleqvtopolsat},
\ref{thm:polsattorespoleqv}~and~\ref{resPolEqvToPolEqv} that
$\PolEqv (S_4)$ is polytime-equivalent to
the the restricted equivalence problem over $\operatorname{Mat}_2(\Z_2)$.
In this section, we provide an equivalent arithmetical problem
over the field $\ob{F}_4$. Here, we see $\ob{F}_4$ as a subring of
$\Mat_2 (\Z_2)$ in the following way: we set $\alpha := \left(\begin{smallmatrix} 0 & 1 \\ 1 & 1\end{smallmatrix}\right)$ and we define
  $\ob{F}_4$ to be the subring $\{0,\alpha,\alpha^2, \alpha^3 = I\}$
  of $\operatorname{Mat}_2(\Z_2)$. The following facts can be verified
  by straightforward calculation:
  \begin{lem} \label{lem:s3}
    Let $\alpha := \left(\begin{smallmatrix} 0 & 1 \\ 1 & 1\end{smallmatrix}\right)$, $\sigma := \left(\begin{smallmatrix} 0 & 1 \\ 1 & 0\end{smallmatrix}\right)$, and let $\GFF = \{0,\alpha,\alpha^2, \alpha^3=I\}$.
        Then we have
        \begin{enumerate}
           \item $\sigma^2 = \alpha^3 = I$,
          \item 
            $\{\sigma \cdot a \mid a \in \GFF \} \cap \GFF  = \{0\}$,
          \item for all $a \in \GFF$ : $\sigma \cdot a \cdot \sigma = a^2$.
        \end{enumerate}
  \end{lem}      
We will now use the field operations of $\Z_3$ and of $\ob{F}_4$, and we write
$+,-,\cdot$ for the operations in $\Z_3$. The addition and subtraction
in $\ob{F}_4$ will be denoted by $\oplus$ and $\ominus$. Of course
$\oplus = \ominus$, but using both symbols will improve readability.
We say that a polynomial $f \in \Z_3 [x_1, \ldots, x_n]$ is \emph{expanded} if
there is an $I \subseteq \{0,1,2\}^n$ such that 
\[
f = \sum_{(e_1, \ldots, e_n) \in I} c_{e_1, \ldots, e_n} x_1^{e_1} \cdots x_n^{e_n}
\]
with $c_{\vb{e}} \in \Z_3 \setminus \{0\}$.
Now the arithmetical problem that we consider is the following:
\begin{prb}[An arithmetic equivalence problem over $\ob{F}_4$] \label{prb:P1} \mbox{} \\
  \textbf{Given:} Expanded polynomials $p_1, \ldots, p_k \in \Z_3 [x_1, \ldots, x_n]$
  with $k,n \in \N$. \\
  \textbf{Asked:} Do we have
  \[
   \forall \vb{a} \in \{-1,1\}^n \, : \, \sum_{i=1}^k \alpha^{p_i (\vb{a})} = 0?
  \]
  Here, the sum is computed in $\ob{F}_4$, hence
  $\sum_{i=1}^k \alpha^{p_i (\vb{a})} = \alpha^{p_1 (\vb{a})} \oplus \cdots \oplus \alpha^{p_k (\vb{a})}$.
\end{prb}
The goal of this section is to show that Problem~\ref{prb:P1} is polytime-equivalent
to $\PolEqv (S_4)$. The next lemma provides a way to compute the product of two matrices in $\Mat_2 (\Z_2)$
inside $\ob{F}_4$:
\begin{lem}\label{lem:rhom}
  On $\ob{F}_4^2$, we define a multiplication $\star$ by
  \[ \VecTwo{a_1}{a_2} \star \VecTwo{b_1}{b_2} := \VecTwoBig{a_1b_1 \, \oplus  \, a_2^2b_2}{a_1^2b_2 \, \oplus \, a_2b_1}
  \text{ for } a_1, a_2, b_1, b_2 \in \GFF,
  \]
  and a mapping
  \[
  \rho: \GFF^2 \rightarrow \Mat_2 (\Z_2), \,  \rho(\VecTwo{a_1}{a_2}) = a_1 \, \oplus \, \sigma \cdot  a_2.
  \]
  Then $\rho$ is an isomorphism from $(\GFF^2, \oplus, \star)$ to
  $(\Mat_2 (\Z_2), +, \cdot)$.
\end{lem}
\begin{proof}
  To see that $\rho$ is injective, we note that for $\vb{a}, \vb{b} \in \GFF^2$ with $\rho(\vb{a}) = \rho(\vb{b})$, we have $0 = \rho(\vb{a}) \ominus \rho(\vb{b}) = (a_1 \ominus b_1) \oplus \sigma (a_2 \ominus b_2)$.
  From Lemma~\ref{lem:s3} , we get $0 = a_1 \ominus b_1 = a_2 \ominus b_2$ and thus $\vb{a} = \vb{b}$.
  Since $|\GFF^2| = |\Mat_2 (\Z_2)|$,  $\rho$ is a bijection.
  One immediately verifies that $\rho$ is a homomorphism from $(\GFF^2, \oplus)$
  to $(\Mat_2 (\Z_2), +)$. We will now show that $\rho$ is also a multiplicative
  homomorphism. To this end, we fix $a_1, a_2, b_1, b_2 \in \ob{F}_4$ and compute,
  using that $\sigma^2 = I$ and $\sigma a = a^2 \sigma$ for all $a \in \GFF$:
  \[
\begin{split}
  \rho (\VecTwo{a_1}{a_2} \star \VecTwo{b_1}{b_2}) &=
  \rho (\VecTwo{a_1 b_1 \oplus a_2^2 b_2}{a_1^2 b_2 \oplus a_2 b_1}) \\
  &= a_1 b_1 \oplus a_2^2 b_2 \oplus \sigma (a_1^2 b_2 \oplus a_2 b_1) \\
  &= a_1 b_1 \oplus a_2^2 b_2 \oplus \sigma a_1^2 b_2 \oplus \sigma a_2 b_1 \\
  &= a_1 b_1 \oplus \sigma a_2 \sigma b_2 \oplus  a_1 \sigma b_2 \oplus \sigma a_2 b_1 \\
    &= a_1 (b_1 \oplus \sigma b_2) \oplus \sigma a_2 (\sigma b_2 \oplus b_1) \\
    &= (a_1 \oplus \sigma a_2)(b_1 \oplus \sigma b_2) \\
    &= \rho (\VecTwo{a_1}{a_2}) \cdot \rho (\VecTwo{b_1}{b_2}).
\end{split}
\]
Hence $\rho$ is an isomorphism from $(\GFF^2, \oplus, \star)$ to
  $(\Mat_2 (\Z_2), +, \cdot)$. 
\end{proof}

Next, we evaluate certain products of $n$ factors in $\GFF^2$. A crucial observation
is that in $\GFF$, the sum of two distinct elements of the set
$\GFF \setminus \{0\} = \{\alpha^0 = I_2, \alpha, \alpha^2\}$ is always equal to
the third element of $\GFF \setminus \{0\}$ distinct from both summands, whereas
$a \oplus a = 0$ for all $a \in \GFF$.
\begin{lem}
  Let $f : \{-1,1\} \times \Z_3\rightarrow \GFF^2$ be defined by
  \[
  f (s, t) = \VecTwo{\alpha^{s+t} \, \oplus \, \alpha^{2+t}}{\alpha^{s+t} \, \oplus \, \alpha^{1+t}}.
  \]
Then $\operatorname{im} (\rho \circ f) = \GL_2(\Z_2)$, and for all $n \in\N, (s_1, t_1), \dots, (s_n, t_n) \in  \{-1,1\} \times \Z_3$, 
we have
  \begin{equation} \label{eq:fst}
    f(s_1, t_1)  \star \dots \star f(s_n, t_n) = f(s_1\cdots s_n, \, \sum_{i = 1}^n ( t_i \! \! \prod_{j=i+1}^n \! \! s_j)).
  \end{equation}  
\end{lem}
\begin{proof}
  For proving the first assertion, we note that
\begin{align*}
\operatorname{im}(\rho \circ f)
&= \rho(\{f(1, t) \mid t \in \Z_3\} \cup \{f(-1,t)\mid t\in\Z_3\})
\\&=\rho(\{(\alpha^t, 0) \mid t \in \Z_3\} \cup \{(0, \alpha^t)\mid t\in\Z_3\})
\\&=\rho(\{(u, v) \in \GFF^2 \mid (u\neq 0 \wedge v = 0) \vee (u = 0 \wedge v \neq 0)\})
\\&= \{1, \alpha, \alpha^2, \sigma,\sigma\alpha, \sigma\alpha^2\}
\\&=\GL_2(\Z_2).
\end{align*}
The last equality holds because the elements of $\{1, \alpha, \alpha^2, \sigma,\sigma\alpha, \sigma\alpha^2\}$ are distinct and invertible and $\GL_2(\Z_2)$ has exactly $6$ elements.
To derive~\eqref{eq:fst}, we first show that for all $(r, u), (s,v) \in\{-1,1\}\times\Z_3$ we have
\begin{equation} \label{eq:fmult}
  f(r,u)\star f(s,v) = f(rs, su + v).
\end{equation}
We distinguish four cases: \\
\textbf{Case} $(r,s) = (1, 1)$: Then $f(1,u) \star f(1,v) =  \VecTwo{\alpha^{u}}{0}
\star \VecTwo{\alpha^{v}}{0} = \VecTwo{\alpha^{u+v}}{0} = f(1, u + v)$.
\textbf{Case} $(r,s) = (1, -1)$: Then $f(1,u) \star f(-1, v) = \VecTwo{\alpha^{u}}{0}
\star \VecTwo{0}{\alpha^{v}} = \VecTwo{0}{\alpha^{-u + v}} =
f (-1, -u  + v)$. \\
\textbf{Case} $(r,s) = (-1, 1)$: Then $f(-1, u) \star f(1, v)) =
\VecTwo{0}{\alpha^{u}} \star \VecTwo{\alpha^{v}}{0} = \VecTwo{0}{\alpha^{u + v}}
= f (-1, u + v)$. \\
\textbf{Case} $(r,s) = (-1, -1)$: $f (-1,u) \star f (-1, v) =
     \VecTwo{0}{\alpha^{u}} \star \VecTwo{0}{\alpha^{v}} 
     = \VecTwo{\alpha^{-u + v}}{0} = f (1, -u + v)$. \\
This completes the proof of~\eqref{eq:fmult}.
Next, we prove \eqref{eq:fst} by induction on $n$. 
For $n = 1$, we have $f(s_1, t_1) = f (s_1, t_1 \cdot 1)$. Since $\prod_{j = 2}^{1} s_j$ is the empty
product, and therefore equal to $1$, we obtain $f (s_1, t_1 \cdot 1) =
f (s_1,  t_1 \prod_{j = 2}^1 s_j)$.
For the induction step from $n$ to $n+1$, we let $n \ge 1$ and 
$(s_1, t_1), \dots, (s_{n+1}, t_{n+1}) \in  \{-1,1\} \times \Z_3$. 
We then have 
\begin{align*}
  f(s_1, t_1) \star & \cdots \star f(s_{n+1}, t_{n+1}) = f(s_1\cdots s_n, \, \sum_{i = 1}^n t_i\prod_{j=i+1}^n s_j)
     \star f(s_{n+1}, t_{n+1}) 
     \\&= f(s_1\cdots s_ns_{n+1},\,  s_{n+1} (\sum_{i = 1}^{n}t_i\prod_{j=i+1}^{n} s_j) + t_{n+1} ) \qquad \text{(by \eqref{eq:fmult})} 
     \\&= f(s_1\cdots s_ns_{n+1},\,  (\sum_{i = 1}^{n} t_i\prod_{j=i+1}^{n+1} s_j) + t_{n+1} )
\\&= f(s_1\cdots s_{n+1}, \, \sum_{i = 1}^{n+1}t_i\prod_{j=i+1}^{n+1} s_j).
\end{align*}
This completes the induction and establishes~\eqref{eq:fst}.
\end{proof}

\begin{thm}
\label{thm:restop1}
The restricted polynomial equivalence problem over $\operatorname{Mat}_2(\Z_2)$ can be polytime-reduced to Problem~\ref{prb:P1}.
\end{thm}
\begin{proof}
  We prepare for the reduction by producing, for each $m \in \N$,
  certain polynomials
  in $\Z_3 [Y]$ in the variables $Y = \{ y_{(i,j)} \mid i \in \ul{m},
  j \in \{1,2,3\}\}$.
  Now we define
  \[
  \begin{array}{rcl}
    u_m &:=& y_{(1,1)} \cdots y_{(m,1)}, \\
    v_m &:=& \sum_{i=1}^m  y_{(i, 2)} \prod_{j=i+1}^m y_{(j, 1)}
            +\sum_{i=1}^m  y_{(i, 3)} \prod_{j=i+1}^m y_{(j, 1)}.
  \end{array}
  \]
  Next, we define expressions $a_m$ and $b_m$ by
  \[
     \begin{array}{rcl}
       a_m &:=& \alpha^{u_m+v_m} \oplus \alpha^{2+v_m}, \\
       b_m &:=& \alpha^{u_m+v_m} \oplus \alpha^{1+v_m}.
     \end{array}
     \]
  Let $f$ be the function defined  in Lemma \ref{lem:rhom}. Then
  for all $\vb{y} \in (\{-1,1\}^3)^m$, we can evaluate
  $a_m (\vb{y})$ and $b_m (\vb{y})$ by
\begin{equation} \label{eq:eval1}
 \VecTwo{a_m(\vb{y})}{b_m(\vb{y})} 
   = \VecTwo{\alpha^{u_m(\vb{y})+v_m(\vb{y})} \, \oplus \, \alpha^{2+v_m(\vb{y})}}
           {\alpha^{u_m(\vb{y})+v_m(\vb{y})} \, \oplus \, \alpha^{1+v_m(\vb{y})}}
           = f( u_m(\vb{y}), v_m(\vb{y})).
\end{equation}
Using the definitions of $u_m$ and $v_m$,
we obtain
\begin{align} \label{eq:eval2}
  f( u_m(\vb{y}),  v_m(\vb{y})) &=
  f(y_{(1,1)}\cdots y_{(m, 1)}, \, \sum_{i = 1}^m (y_{(i, 2)} + y_{(i, 3)}) \!
  \prod_{j=i+1}^m y_{(j, 1)}).
\end{align}
By~\eqref{eq:fst}, the last expression is equal to
\[
    f(y_{(1,1)}, y_{(1, 2)} + y_{(1,3)}) \star \dots  \star f(y_{(m,1)}, y_{(m, 2)} + y_{(m,3)}).
\]

We will now start the reduction of the restricted equivalence problem.
To this end, let $n\in \N$ and let $p$ be a restricted polynomial expression in the variables $x_1, \dots, x_n$. Then $p$ is a sum of restricted monomial
expressions, hence there are $k\in\N_0$ and words $p_1, \dots, p_k \in (\GL_2(\Z_2) \cup \{x_1, \dots, x_n\})^\ast$ such that
\[ p = p_1 + \cdots + p_k. \]
We will first transform $p$ into an expression $q$ whose variables
$\bigcup_{j \in\underline{n}} \{x_{(j, 1)}, x_{(j, 2)}, x_{(j, 3)}\}$ 
are intended to range over $\{-1, 1\}$, and whose value
lies in $\GFF^2$. The fundamental property
of this expression $q$ will be that
for all $((x_{(j, 1)}, x_{(j, 2)}, x_{(j, 3)}))_{j \in\underline{n}} \in
(\{-1,1\}^3)^n$, we have
\begin{multline} \label{eq:pq}
  p (\rho \circ f (x_{(1,1)}, x_{(1,2)} + x_{(1,3)}), \ldots,
     \rho \circ f (x_{(n,1)}, x_{(n,2)} + x_{(n,3)})) \\ =
  \rho (q ( x_{(1, 1)}, x_{(1, 2)}, x_{(1, 3)}, \ldots,
            x_{(n, 1)}, x_{(n, 2)}, x_{(n, 3)})).
\end{multline}  
To this end, we let $i \in \ul{k}$, and we apply the following
rewrite rules to $p_i$, which we assume to have length $m$.
First, the variables in the word $p_i$ get replaced
using the rule
\begin{equation*}
    x_k  \longrightarrow  f (x_{(k,1)}, x_{(k,2)} + x_{(k,3)})
\end{equation*}
for $k \in \ul{n}$. For each constant 
$c \in \GL_2(\Z_2)$, we find $s(c, 1), s(c, 2)$ and $s(c, 3) \in \{-1, 1\}$
with
$\rho(f(s (c, 1), s(c, 2) + s(c, 3)) = c$.
Such $s(c,1), s(c,2)$ and $s(c,3)$ exists for all $c \in \GL_2(\Z_2)$
because $\operatorname{im}(\rho \circ f) = \GL_2(\Z_2)$ and because for each $u \in \Z_3$ there are $s, t\in\{-1,1\}$ with $u = s+t$.
Then the replacement
of the constants is done by applying the rewrite rules
\begin{equation*}
  c \longrightarrow f(s (c, 1), s(c, 2) + s(c,3))
\end{equation*}
with $c \in \GL_2 (\Z_2)$.
Applying these rewrite rules, the word $p_i$ over the alphabet
$\GL_2 (\Z_2) \cup \{x_1, \ldots, x_n\}$ is transformed into
a word over the alphabet
$\{ f(1,1+1)), f(1,1 + (-1)), \ldots, f((-1),(-1) + (-1)\} \cup
\{ f (x_{(k,1)}, x_{(k,2)} + x_{(k,3)}) \mid k \in \ul{n} \}$.
This word is of the form
\[
f(z_{1,1}, z_{1,2} + z_{1,3}) \, \ldots \, f (z_{m,1}, z_{m,2} + z_{m,3}),
\]
where all $z_{j,i}$ are elements of $\{1,-1\} \cup \{ x_{(k,l)} \mid
   (k, l) \in \ul{n} \times \{1,2,3\}$.
   We transform this word into the expression $q_i$ given by
\[
q_i := f(z_{1,1}, z_{1,2} + z_{1,3}) \star \cdots \star
f (z_{m,1}, z_{m,2} + z_{m,3}).
\]
We see that these replacements guarantee
\begin{multline*}
p_i (\rho \circ f (x_{(1,1)}, x_{(1,2)} + x_{(1,3)}), \ldots,
     \rho \circ  f (x_{(n,1)}, x_{(n,2)} + x_{(n,3)})) \\ =
  \rho (q_i ( x_{(1, 1)}, x_{(1, 2)}, x_{(1, 3)}, \ldots,
  x_{(n, 1)}, x_{(n, 2)}, x_{(n, 3)}))
\end{multline*}
  for all $((x_{(j, 1)}, x_{(j, 2)}, x_{(j, 3)}))_{j \in\underline{n}} \in
  (\{-1,1\}^3)^n$.
  Next, we compute an expression $q_i'$ from $q_i$ by defining
  \[
  q_i' = \VecTwoBig{a_m ( (z_{1,1}, z_{1,2}, z_{1,3}), \ldots,
    (z_{m,1}, z_{m,2}, z_{m,3}) )}
                {b_m ( (z_{1,1}, z_{1,2}, z_{1,3}), \ldots,
                  (z_{m,1}, z_{m,2}, z_{m,3}) )}.
  \]              
  The calculations in \eqref{eq:eval1} and~\eqref{eq:eval2}
  justify that
  \begin{multline*}
    q_i ( x_{(1, 1)}, x_{(1, 2)}, x_{(1, 3)}, \ldots,
  x_{(n, 1)}, x_{(n, 2)}, x_{(n, 3)}) \\ =
   q_i' ( x_{(1, 1)}, x_{(1, 2)}, x_{(1, 3)}, \ldots,
   x_{(n, 1)}, x_{(n, 2)}, x_{(n, 3)})
  \end{multline*}
   for all $((x_{(j, 1)}, x_{(j, 2)}, x_{(j, 3)}))_{j \in\underline{n}} \in
   (\{-1,1\}^3)^n$.
   Finally, we compute
   \[
   q = q_1' \oplus \cdots \oplus q_k',
   \]
   and see that $q$ satisfies the property~\eqref{eq:pq}.
   Now $q' = \VecTwo{a}{b}$, where $a,b$ are expressions
   whose variables range over $\{-1,1\}$ that evaluate
   to elements of $\GFF$. Let $x_{(n+1, 1)}$ be a new variable.
   We define an expression $r$
   by
   \[
   r := (\alpha^{x_{(n+1,1)}} \oplus \alpha) \cdot a \,\, \oplus \,\,
   (\alpha^{x_{(n+1,1)}} \oplus \alpha^2) \cdot b,
   \]
   and we let $r'$ its expansion using the laws
   $(\beta \oplus \gamma) \cdot \delta = \beta \cdot \delta \,\oplus\,
   \gamma \cdot \delta$
   and $\alpha^{u} \cdot \alpha^{v} = \alpha^{u+v}$ of $\GFF$.
   
   Next we show that $r'$ evaluates to $0$ for all
   $\vb{a} \in \{-1,1\}^{3 n + 1}$ if and only if
   $p (\vb{b}) = 0$ for all $\vb{b} \in \GL_2 (\Z_2)^n$.
   Let us first assume that $p (\vb{b}) \neq 0$ with
   $\vb{b} = (b_1, \ldots, b_n) \in \GL_2 (\Z_2)^n$. Then
   we find $b_{(i,j)}'s$ in $\{-1,1\}$ with
   $f(b_{(i,1)}, b_{(i,2)} + b_{(i,3)}) = b_i$ for $i \in \ul{n}$.
   By~\eqref{eq:pq}, we find that
   $q ( b_{(1, 1)}, b_{(1, 2)}, b_{(1, 3)}, \ldots,
   b_{(n, 1)}, b_{(n, 2)}, b_{(n, 3)})) =: \VecTwo{a'}{b'}
   \neq \VecTwo{0}{0}$.
   If $a' \neq 0$, then we set $x_{(n+1,1)} := -1$ and obtain a
   place at which $r$ is not zero; if $b' \neq 0$, we set
   $x_{(n+1,1)} := 1$ and obtain a place at which $r$ is not zero.

   Now we assume that $p (\vb{b}) = 0$ for all $\vb{b} \in \GL_2 (\Z_2)$.
   Then by~\eqref{eq:pq}, we have that $q$ is zero for all values
   in $(\{-1,1\}^3)^n$, which clearly implies that $r'$ is zero
   for all $\vb{a} \in \{-1,1\}^{3n+1}$.
   \end{proof}

\begin{thm}
    \label{thm:back2}
    Problem~\ref{prb:P1}
    can be polytime-reduced to the restricted equivalence
    problem over $\operatorname{Mat}_2(\Z_2).$
\end{thm}
\begin{proof}
  Let $p_1, \ldots, p_k$ be expanded polynomials in
  $\Z_3 [x_1, \ldots, x_n]$. We seek to determine
  whether
  $\forall \vb{x} \in \{-1,1\}^n : \sum_{i=1}^k \alpha^{p_i (\vb{x})} = 0$.
  Replacing every monomial of the form $2x_1^{e_1}\cdots x_n^{e_n}$
  by $x_1^{e_1}\cdots x_n^{e_n} + x_1^{e_1}\cdots x_n^{e_n}$, we
  may assume that each $p_i$ is of the form
  \[
  p_i = m_{i,1} + \dots + m_{i,r_i}, 
  \]
  where each $m_{i,j}$ 
  is a product of variables from $x_1, \dots, x_n$.
  Next, we define a function $\sgn : \GL_2 (\Z_2) \to \{-1,1\}$
  as follows: every $x \in \GL_2(\Z_2)$ can be written
  $x = \sigma^i\alpha^j$ with $i \in\{0,1\}$ and
  $j \in \{0,1,2\}$. Then we define $\sgn(x) := (-1)^i$.
  We note that $\sgn (x)$ is just the signature of the
  permutation $\pi (x)$ when $\pi$ is an isomorphism
  from $\GL_2 (\Z_2)$ to the symmetric group $S_3$.
  Then for all $x, y \in \GL_2(\Z_2)$, we have
  $$\sgn(xy) = \sgn(x)\sgn(y)$$ and
   $$x\alpha x^{-1} = \sigma^i\alpha^j \alpha \alpha^{-j}\sigma^{-i} = \alpha^{\sgn(x)}.$$
  We will now define a restricted polynomial expression
  over $\matrixRing$ by
  \[
  s := \sum_{i\in\underline{k}} \prod_{j\in\underline{r_i}}
  m_{i,j} \, \alpha \, m_{i,j} \, m_{i,j} \,  m_{i,j} \, m_{i,j} \, m_{i,j}.
  \]
  We will now show that 
  for all $\vb{y} \in \GL_2(\Z_2)^n$, we have
  \begin{equation} \label{eq:sp}
    s(\vb{y}) = p(\sgn(y_1), \dots, \sgn(y_n)).
  \end{equation}
  To this end, let $\vb{y} \in \GL_2(\Z_2)^n$.
  We have
  \[
  \begin{split}
    s (\vb{y})  & =
      \sum_{i\in\underline{k}} \prod_{j\in\underline{r_i}}
      m_{i,j} (\vb{y}) \cdot \alpha \cdot (m_{i,j} (\vb{y}))^{-1}
      \\
      & =
      \sum_{i\in\underline{k}} \prod_{j\in\underline{r_i}}
      \alpha^{\sgn (m_{i,j} (\vb{y}))}.
  \end{split}
  \]
  Since $\sgn$ is a homomorphism from $(\GL_2 (\Z_2), \cdot)$
  to $(\Z_3 \setminus \{0\}, \cdot)$,
  the last expression is equal to
  \[
  \begin{split}
     \sum_{i\in\underline{k}} \prod_{j\in\underline{r_i}}
     \alpha^{m_{i,j} (\sgn (y_1), \ldots, \sgn (y_n))}
        & = 
      \sum_{i\in\underline{k}} 
      \alpha^{\sum_{j\in\underline{r_i}}m_{i,j} (\sgn (y_1), \ldots, \sgn (y_n))}
      \\
      & =
          \sum_{i\in\underline{k}} 
          \alpha^{p_i (\sgn (y_1), \ldots, \sgn (y_n))} \\
          & = p (\sgn (y_1), \ldots, \sgn (y_n)).
  \end{split}
 \] 
This proves~\eqref{eq:sp}. From this equation, it follows that 
$\forall \vb{y}\in \{-1,1\}^n: p(\vb{y}) = 0$ holds if and only if $\forall \vb{y}\in\GL_2(\Z_2)^n: s(\vb{y}) = 0$.
\end{proof}

\section{Reduction to modular circuits} \label{sec:circuits}
In this section, we relate Problem~\ref{prb:P1} to the circuit equivalence problem for $\operatorname{CC}[2,3,2]$-circuits. We will express such a circuit using
the functions $\operatorname{Mod}_n : \Z \to \Z$ defined
by $\operatorname{Mod}_n (x) = 1$ if $n$ divides $x$ and
$\operatorname{Mod}_n (x) = 0$ if $n \nmid x$.
Since the depth of the circuits we consider is constant, we can transform
$\operatorname{CC}[2,3,2]$-circuits
into nested expressions with polynomial overhead (see \cite[Section~3]{IKK:IPIM} for such a transformation in a slightly different formalism).

One example of an expression that represents a $\operatorname{CC}[2,3,2]$-circuit is
\[
C(y_1, y_2, y_3) := \operatorname{Mod}_2\Big(1 + \operatorname{Mod}_3\big(1 + 2\cdot \operatorname{Mod}_2(y_1 + y_2)\big) + \operatorname{Mod}_3\big(\operatorname{Mod}_2(y_3)\big)\Big).
\]
\begin{lem}
\label{lem:d2}
Let $n\in\N$ and let $p\in\Z_3[x_1, \dots, x_n]$ be an expanded polynomial. 
Then we can construct a $\operatorname{CC}[3,2]$-circuit $C$ of size $O(\length{p})$ such that for all $\vb{x}\in \{0,1\}^n,$ we have $p((-1)^{x_1}, \dots, (-1)^{x_n}) = 0$ if and only if $C(\vb{x}) = 1$.
\end{lem}
\begin{proof}
  Since we are only interested in the values of $p$ at $\{-1,1\}$, we may
  replace $p$ by its remainder modulo  $\{x_i^2 - 1 \mid i \in \ul{n} \}$.
  This remainder $p'$ has degree at most $1$ in every variable. 
  Now choose $E \subseteq \{0,1\}^n$ and $c: E \rightarrow \Z_3\setminus \{0\}$ such that $p' = \sum_{e\in E} c(e) \vb{x}^e$, where $\vb{x}^e$ is a shorthand
  of $x_1^{e_1} \cdots x_n^{e_n}$.
Consider the circuit 
$$
C(\vb{y}) := \operatorname{Mod}_3\big( \sum_{e\in E} c(e) (-1 + 2\cdot \operatorname{Mod}_2(\sum_{i=1}^n y_ie_i))\big).
$$
Now let $\vb{y} \in \{0,1\}^n$ and let $\vb{x} := ((-1)^{y_1}, \ldots,
    (-1)^{y_n}))$.
For each $e\in E,$ we have $c(e)\vb{x}^e = c(e)(-1)^{\sum_{i=1}^n y_ie_i} = c(e) \big(-1 + 2\cdot \operatorname{Mod}_2(\sum_{i=1}^n y_ie_i)\big)$, and therefore
$p(\vb{x}) = 0$ if and only if $\sum_{e\in E} c(e) (-1 + 2\cdot \operatorname{Mod}_2(\sum_{i=1}^n y_ie_i)) \equiv_3 0.$
\end{proof}

\begin{thm} 
\label{thm:circ}
Problem~\ref{prb:P1} can be polytime-reduced to the circuit equivalence problem for $\operatorname{CC}[2,3,2]$-circuits.
\end{thm}
\begin{proof}
Let $Q(\vb{x}) := \sum_{i=1}^k \alpha^{p_i(\vb{x})}$ be an instance to Problem~\ref{prb:P1}, where $p_1,\dots, p_k\in \Z_3[x_1, \dots, x_n]$.
For every expanded polynomial $p \in \Z_3 [x_0, \ldots, x_n]$, let
$C(p)$ be the  $\operatorname{CC}[3,2]$-circuit with input wires $y_0, \dots, y_n$ produced from Lemma~\ref{lem:d2}. Hence for all $\vb{y} \in \{0,1\}^n$,
  we have
  \begin{equation} \label{eq:pc}
    p( (-1)^{y_1}, \ldots, (-1)^{y_n} ) = 0 \Longleftrightarrow
    C(p) \, (y_1, \ldots, y_n) = 1.
  \end{equation}
We define $f(a,b) := (b + 1)a(a - 2) + (b - 1)(a - 1)(a - 2).$ 
Note that $f(a,1) = 0$ if and only if $a\in \{0, 2\}$ and $f(a,-1) = 0$ if and only if $a \in \{1,2\}.$
Now consider the circuit
$$
C_Q(y_0,\dots, y_n) := \operatorname{Mod}_2(\sum_{i=1}^k C(f(p_i, x_0))).
$$
We will prove that for all $\vb{y}\in \{0, 1\}^n$, we have
\begin{equation} \label{eq:qc}
  Q((-1)^\vb{y}) = 0 \Longleftrightarrow C_Q(0, \vb{y}) = C_Q (1, \vb{y}) = 1.
\end{equation}
Here $(-1)^{\vb{y}}$ is a shorthand for $(-1)^{y_1} \cdots (-1)^{y_n}$.
Let $\vb{y}\in \{0, 1\}^n$. 
For $a\in \Z_3$, we set $k_a := \setsize{\{i \in \ul{k} : p_i((-1)^\vb{y}) = a\}}$.
Clearly
\[ Q((-1)^\vb{y}) = k_0 \cdot 1 + k_1 \cdot \alpha + k_2 \cdot (1 + \alpha) = (k_0 + k_2)\cdot 1 + (k_1 + k_2)\cdot \alpha.
\]
Therefore $Q((-1)^\vb{y}) = 0$ if and only if $k_0 + k_2 \equiv_2 0$ and $k_1 + k_2 \equiv_2 0$.
This is the case if and only if $0 \equiv_2 \setsize{\{i \in \ul{k} : f(p_i((-1)^\vb{y}), 1) = 0\}}$ and $0 \equiv_2 \setsize{\{i \in \ul{k}: f(p_i((-1)^\vb{y}), -1) = 0\}}$.
By~\eqref{eq:pc}, for each $i \in \ul{k}$, we have
$p_i ((-1)^\vb{y}, 1) = 0$ iff $C (f(p_i,x_0)) \, (0, \vb{y}) = 1$ and
$p_i ((-1)^\vb{y}, -1) = 0$ if $C (f(p_i,x_0)) \, (1, \vb{y}) = 1$.
Hence $Q((-1)^\vb{y}) = 0$ holds if and only if
$C_Q(0,\vb{y}) = C_Q(1, \vb{y}) = 1$.
\end{proof}

\begin{thm}
    \label{thm:back1}
    The circuit equivalence problem for $\operatorname{CC}[2,3,2]$-circuits can be polytime-reduced to Problem~\ref{prb:P1}.
\end{thm}
\begin{proof}
    Let 
    \[
        C(\vb{y}) = \operatorname{Mod}_2\Big(c_0 + \sum_{i=1}^k
        \operatorname{Mod}_3\Big(c_1 (i) + \sum_{j=1}^{l(i)} \operatorname{Mod}_2\big(c_2 (i,j) + \sum_{m=1}^n e (i,j, m) y_m \big)\Big)
        \Big)
    \]
    be a $\operatorname{CC}[2,3,2]$-circuit with $n$ input wires $y_1, \ldots, y_n$,
    where $c_0, c_2(i,j), e(i,j, m) \in \{0,1\}$ and
    $c_1 (i) \in \{0,1,2\}$ for $1 \leq i \leq k$, $1 \leq j \leq l(i)$, $1 \le m \le n$.

    Let $r(z) := (z \ominus \alpha)(z \ominus \alpha^2)$,
    Writing $\alpha \uparrow x$ for $\alpha^x$, we
    consider the expression $Q (\vb{x})$ obtained
    by expanding
    \[
        c_0 + \sum_{i = 1}^k r \Big(\alpha \uparrow
        \big( c_1(i) + \sum_{j = 1}^{l(i)} -1 - (-1)^{c_2(i,j)}\vb{x}^{\vb{e}(i,j)}\big)\Big).
    \]
    Now let $\vb{y}\in \{0,1\}^n$ and let $\vb{x} := (-1)^\vb{y}$.
    We will show that $C(\vb{y}) = 1$ if and only if $Q(\vb{x}) = 0$.
    For $i\in \{1, \dots, k\}$, we have
    \begin{equation} \label{eq:m3}
        \operatorname{Mod}_3\Big(c_1 (i) + \sum_{j=1}^{l(i)} \operatorname{Mod}_2\big(c_2 (i,j) + \sum_{m=1}^n e (i,j, m) y_m \big)\Big) = 1
    \end{equation} 
    if and only if
    $$
    c_1(i) + \sum_{j = 1}^{l(i)} -1 -  (-1)^{c_2(i,j)}\vb{x}^{\vb{e}(i,j)} \equiv_3 0
    $$
    if and only if
    \[
        r \Big(\alpha \uparrow
        \big( c_1(i) + \sum_{j = 1}^{l(i)} -1 - (-1)^{c_2(i,j)}\vb{x}^{\vb{e}(i,j)}\big)\Big) = 1.
    \]
    Hence, if~\eqref{eq:m3} does not hold, then since
    the range of $r$ is $\{0, \alpha^0\}$,
    we have $r \Big(\alpha \uparrow
    \big( c_1(i) + \sum_{j = 1}^{l(i)} -1 - (-1)^{c_2(i,j)}\vb{x}^{\vb{e}(i,j)}\big)\Big) = 0$.
    Therefore $C(\vb{y}) = 1$ if and only if $Q(\vb{x}) = 0.$
\end{proof}

This theorem also completes the proof of Theorem~\ref{thm:S4}. With these reductions,
we can also prove~Theorem \ref{thm:algo}. 
\begin{proof}[Proof of Theorem \ref{thm:algo}]
We adapt the proof of \cite[Theorem~6.1]{IKK:COMC}.
Theorem~\ref{thm:S4} yields a $d_1>0$ and an algorithm $A$ that takes a polynomial $p$ over $S_4$ and returns,
in time $\le |p|^{d_1}$, a $\operatorname{CC}[2,3,2]$-circuit $C'$ such that $p(x_1, \dots, x_n)=1$ is satisfiable if and only if $C'(y_1, \dots, y_{m}) \neq 1$ is satisfiable. Flipping the output of
this circuit, which can be achieved by adding a constant bit at the last level, we obtain
a circuit $C = \operatorname{Mod}_2 (1 + C')$ such that $p(x_1, \ldots, x_n) = 1$ is satisfiable
if and only if $C (y_1, \ldots, y_m) = 1$ is satisfiable.

Based on this algorithm $A$, we will now construct algorithms deciding
$\PolSat (S_4)$. To this end, we show that 
$C=1$ has a solution if and only if it has a solution with at most $\gamma^{-1}(|C|)$ nonzero coordinates:
If $C$ always evaluates to $1$, this is clear.
Otherwise let $z$ be a solution to $C(z) = 1$ with $k:= \setsize{\{i: z_i=1\}}$ minimal.
Let $C'$ be the circuit obtained by replacing the input $i$ with the constant $0$ for all $i\in\ul{m}$ with $z_i=0.$
By minimality of $z,$ this circuit $C'$ computes $\text{AND}_k$.
But then $|C| \geq |C'| \geq \gamma(k)$ and thus $k \leq \gamma^{-1}(|C|)$

The deterministic algorithm is now as follows:
\begin{itemize}
    \item Step 1: Call $A$ to obtain $C$ as described above.
    \item Step 2: Check all tuples $y\in \{0,1\}^m$ with at most $\gamma^{-1}(|C|)$ nonzero entries and return true if $C(y) = 1$ for one of them. Otherwise return false.
\end{itemize}
This algorithm is correct by the above argument.
Now let $n = |p|$.
By assumption, Step 1 takes time
at most $n^{d_1}$; then we also have $m \le n^{d_1}$. There is $d_2 \in \N$ such that
every evaluation of $C$ in Step 2 takes time at most $m^{d_2} \le n^{d_1 d_2}$.
There are at most $2^{\gamma^{-1}(|C|)} \cdot  m^{\gamma^{-1}(|C|)}$ points to check. 
Therefore we can bound the runtime by
$n^{d_1} + n^{d_1 d_2} 2^{\gamma^{-1} (|C|)} n^{d_1 \gamma^{-1} (|C|)}$.
Since $|C| \le n^{d_1}$, the runtime depending on $n$ is a function in 
$\exp(O(\log(n)\cdot \gamma^{-1}(n^{d_1})))$.

Next we define the probabilistic algorithm.
We first show via induction on $m$ that for a circuit $D$ in $m$ variables, if $D(y) = 1$ has a solution then $\#\{y \mid D(y) = 1\} \geq 2^{m-\gamma^{-1}(|D|)}.$ 
For $m=0,$ this follows from $\gamma^{-1}(|D|) \geq 0.$
For the induction step, we distinguish two cases.
In the case that $D(y)=1$ has exactly one solution, $D$ computes $\text{AND}_m$, thus $|D| \geq \gamma(m),$ hence $m \leq \gamma^{-1}(|D|)$.
    Therefore $2^{m-\gamma^{-1}(|D|)} \leq 1$ is a lower bound on the number of solutions.
In the case that $D(y)=1$ has at least two distinct solutions, we assume, 
without loss of generality, that there are two solutions that disagree at index $m$.
Now for $a\in \{0,1\}$ consider the circuits $D_a(y_1, \dots, y_{m-1}) := D(y_1, \dots, y_{m-1}, a).$
For $a\in \{0,1\}$ let $S_a := \{y\in \{0,1\}^m: D(y) = 1 \land y_m=a\}.$
By the induction hypothesis, $\setsize{S_a} \geq 2^{m-1-\gamma^{-1}(|D_a|)} \geq 2^{m-1-\gamma^{-1}(|D|)}.$
Since the set $S:= \{y\in \{0,1\}^m \mid D(y) = 1\}$ is the disjoint union of $S_0$ and $S_1,$ its size is
 $\setsize{S_0} + \setsize{S_1} \geq 2^{m-\gamma^{-1}(|D|)}.$
This implies that a randomly chosen $y \in \{0,1\}^m$ is a solution to $D (y) = 1$ with
probability $p' > 2^{- \gamma^{-1} (|D|)}$.
    The probabilistic algorithm then works as follows:
    \begin{itemize}
        \item Step 1: Obtain $C$ as described above.
        \item Step 2: Generate vectors $y_1, \dots, y_k\in \{0,1\}^m$ at random (independently and uniformly
          distributed), where $k = 2^{\gamma^{-1}(|C|)}$.
        \item Step 3: Return true if $C(y_j)=1$ for some $j$ and false otherwise.
    \end{itemize}
    If there is no solution to $C(y) = 1$, then the algorithm will always give the correct answer
    ``false''.
    If $C(y) = 1$ has a solution, we let $g := 2^{\gamma^{-1}(|C|)}$. Then the probability for
    the wrong answer ``false'' is at most $(1 - \frac{1}{g})^g \le \exp (-1) < \frac{1}{2}$. 
    The runtime of the algorithm is $\exp(O(\log(n) + \gamma^{-1}(n^d))).$
\end{proof}

\newcommand{\etalchar}[1]{$^{#1}$}
\def\cprime{$'$}
\providecommand{\bysame}{\leavevmode\hbox to3em{\hrulefill}\thinspace}
\providecommand{\MR}{\relax\ifhmode\unskip\space\fi MR }
\providecommand{\MRhref}[2]{%
  \href{http://www.ams.org/mathscinet-getitem?mr=#1}{#2}
}
\providecommand{\href}[2]{#2}

 \end{document}